%% file: SICON_main.tex
\documentclass[preprint]{siamart251216}
\usepackage[numbers]{natbib}
\usepackage{tikz}
\usepackage{graphicx} 
\usepackage{enumitem}
\input{Sections/notation}
\newenvironment{talign*}
  {\begingroup\let\displaystyle\textstyle
   \begin{equation*}\begin{aligned}}
  {\end{aligned}\end{equation*}\endgroup}

\input{ex_shared}

\ifpdf
\hypersetup{
  pdftitle={Equilibrium Selection for Multi-agent Reinforcement Learning: A Unified Framework},
  pdfauthor={}
}
\fi

\allowdisplaybreaks

\usepackage[acronym]{glossaries}
\newacronym{acr:rl}{RL}{Reinforcement Learning}
\newacronym{acr:marl}{MARL}{Multi-Agent RL}
\begin{document}

\maketitle
\vspace{-5pt}
\begin{abstract}
 While multi-agent reinforcement learning (MARL) has produced numerous algorithms that converge to Nash or related equilibria, such equilibria are often non-unique and can exhibit widely varying efficiency.
 This raises a fundamental question: how can one design learning dynamics that not only converge to equilibrium but also select equilibria with desirable performance, such as high social welfare?
 In contrast to the MARL literature, equilibrium selection has been extensively studied in normal-form games, where decentralized dynamics are known to converge to potential-maximizing or Pareto-optimal Nash equilibria (NEs).
 Motivated by these results, we study equilibrium selection in finite-horizon stochastic games. 
 We propose a unified actor-critic framework in which a critic learns state-action value functions, and an actor applies a classical equilibrium-selection rule state-wise, treating learned values as stage-game payoffs.
 We show that, under standard stochastic stability assumptions, the stochastically stable policies of the resulting dynamics inherit the equilibrium selection properties of the underlying normal-form learning rule.
 As consequences, we obtain potential-maximizing policies in Markov potential games and Pareto-optimal (Markov perfect) equilibria in general-sum stochastic games, together with sample-based implementation of the framework.
 \end{abstract}
\vspace{-5pt}
\vspace{-5pt}

\input{Sections/introduction}

\input{Sections/prelim}

\input{Sections/unified_framework}
\input{Sections/proof_sketches}

\input{Sections/numerics}
\input{Sections/conclusion}
\appendix
\input{Sections/appendix_proofs}
\input{Sections/appendix_numerics}

\bibliographystyle{unsrtnat}
\bibliography{bib}
\end{document}

%% file: Sections/notation.tex
\newcommand{\cA}{\mathcal{A}}

\newcommand{\cE}{\mathcal{E}}
\newcommand{\cS}{\mathcal{S}}
\newcommand{\cM}{\mathcal{M}}

\newcommand{\bR}{\mathbb{R}}

\newcommand{\C}{\textup{C}}
\newcommand{\D}{\textup{D}}

\newcommand{\bars}{\overline{s}}
\newcommand{\bara}{\overline{a}}

\DeclareMathOperator*{\argmin}{arg\,min}
\DeclareMathOperator*{\argmax}{arg\,max}
\newcommand{\bE}{\mathbb{E}}
\newcommand{\asto}{\xrightarrow{\text{a.s.}}}
\renewcommand{\^}[1]{^{(#1)}}
\newcommand{\brac}[1]{{\{#1\}_{i=1}^n}}
\newcommand{\pieps}{{\pi^\epsilon}}
\newcommand{\Peps}{{K^\epsilon}}

\newcommand{\cX}{\mathcal{X}}
\newcommand{\tmix}{t_{\textup{mix}}}
\newcommand{\taumix}{\tau_{\textup{mix}}}

\renewtheorem{theorem}{Theorem}
\makeatletter
\@removefromreset{theorem}{section}
\makeatother
\renewtheorem{lemma}{Lemma}
\newtheorem{defi}{Definition}
\newtheorem{coro}{Corollary}
\newtheorem{assump}{Assumption}
\newtheorem{example}{Example}
\newtheorem{coroexmp}[coro]{Example and Corollary}
\newcommand{\cT}{\mathcal{T}}
\newcommand{\gammaeps}{{\gamma^\epsilon}}

\newcounter{rmk}
\renewcommand{\thermk}{\arabic{rmk}}

\newenvironment{rmk}[1][]%
{%
  \refstepcounter{rmk}%
  \par\medskip\noindent
  \textbf{Remark~\thermk.}%
  \if\relax\detokenize{#1}\relax\else\ \textbf{(#1)}\fi
  \ \normalfont
}%
{\par\medskip}

%% file: ex_shared.tex
\usepackage{lipsum}
\usepackage{amsfonts}
\usepackage{graphicx}
\usepackage{epstopdf}
\usepackage{algorithmic}
\ifpdf
  \DeclareGraphicsExtensions{.eps,.pdf,.png,.jpg}
\else
  \DeclareGraphicsExtensions{.eps}
\fi

\newsiamremark{remark}{Remark}
\newsiamremark{hypothesis}{Hypothesis}
\crefname{hypothesis}{Hypothesis}{Hypotheses}
\newsiamthm{claim}{Claim}

\headers{Equilibrium Selection for Multi-agent Reinforcement Learning}{R. Zhang, G. Zardini, A. Ozdaglar, J. Shamma, N. Li}

\title{Equilibrium Selection for Multi-agent Reinforcement Learning: A Unified Framework}

\author{Runyu Zhang\thanks{Laboratory for Information and Decision Systems, Massachusetts Institute of Technology,  (\email{runyuzha@mit.edu}, \email{gzardini@mit.edu}, \email{asuman@mit.edu}).}
  \and Gioele Zardini\footnotemark[1]
  \and Asuman Ozdaglar\footnotemark[1]
\and Jeff Shamma\thanks{UIUC, Department of Industrial and Enterprise Systems Engineering
  (\email{jshamma@illinois.edu}}).
\and Na Li\thanks{Harvard University, School of Engineering and Applied Sciences, (\email{nali@seas.harvard.edu}})}
\usepackage{amsopn}


%% file: Sections/introduction.tex
\vspace{-5pt}
\section{Introduction}
Multi-agent systems arise in numerous real-world applications, including power and energy management \citep{lagorse2010multi}, network routing \citep{claes2011}, and robotic swarms \citep{liu2018}. 
Advances in reinforcement learning (RL) have led to renewed interest in multi-agent reinforcement learning (MARL), and stochastic games (SGs) have become a standard modeling framework for studying such problems.
Since Shapley's introduction of stochastic games~\cite{shapley53}, a large literature has studied equilibrium-seeking algorithms in SGs and MARL (see, e.g., \cite{Littman94, Bowling00, Shoham03, Bucsoniu10, Lanctot17, Zhang19}), and more recent work has begun providing finite-time and sample-complexity guarantees for learning solution concepts such as Nash equilibria (NE), correlated equilibria (CE), and coarse correlated equilibria (CCE).
For example, \cite{zhang2021gradient,leonardos2021global,Zhang2022softmax} study convergence to Nash equilibrium for Markov potential games. \cite{bai2020,daskalakis2021,Zhang2022policy} focus on NE-seeking algorithm under the two-player zero-sum setting. \cite{song2021,jin2021v} study efficient learning and convergence to CE and CCE for the general-sum game.

A fundamental challenge in SGs (and more generally in multi-agent systems) is that equilibria are typically not unique and their social performance can differ dramatically. 
Most MARL research focuses on designing algorithms that converge to \emph{an} equilibrium, or on bounding worst-case performance via measures such as the price of anarchy (PoA)~\citep{zhang2023PoA, chen2022convergence}.
These approaches, however, do not address the complementary, and practically important, question of how to select or design dynamics that reliably lead to high-quality equilibria rather than poor ones.


In contrast, equilibrium selection has been extensively studied in the normal-form game literature.
A prominent approach is stochastic stability: when players occasionally make small mistakes, only a subset of equilibria remain robust in the long run;
these are the stochastically stable equilibria (SSE) studied in~\cite{foster1990stochastic, young_evolution_1993} and subsequent work \citep{samuelson1997evolutionary, fudenberg2009reputation, FRANKEL20031}.
Crucially, by designing the players' learning dynamics one can influence which equilibria are stochastically stable: for instance, log-linear learning selects global maxima of a potential function in potential games~\citep{blume_statistical_1993, marden_revisiting_2012}, while specially crafted dynamics can ensure that only Pareto-optimal equilibria are stochastically stable in general-sum games \citep{pradelski2012learning, marden_achieving_2012}.
A rich body of additional work has analyzed related dynamics (fictitious play, joint strategy fictitious play, adaptive play, trial-and-error learning) and demonstrate practical applications in domains such as wind-farm control, target assignment, and routing~\citep{Monderer96fictitious, marden2009joint, young2009learning, young2020individual, arslan2007autonomous, marden2009payoff}.

\vspace{2pt}
\noindent\textbf{Our contributions.} 
We introduce a unified, modular framework that systematically imports equilibrium-selection ideas from normal-form games into MARL.
Concretely, we design an actor-critic architecture in which the critic estimates stage-wise~$Q$-functions of the stochastic game and the actor runs a classical normal-form equilibrium-selection rule using those~$Q$-values as stage payoffs.
This simple but powerful reduction lets us ``lift'' a wide family of well-studied normal-form dynamics to stochastic games and, crucially, transfer their stochastic-stability guarantees accordingly.
Under standard ergodicity and resistance-structure assumptions, we characterize exactly which Markov policies become stochastically stable: for Markov potential games, embedding log-linear learning yields only potential-maximizing policies in the zero-noise limit;
for general-sum games, suitably designed dynamics produce only utility-maximizing (Pareto-optimal) stochastically stable policies.
Our framework is \emph{algorithmic} (we give both a model-based critic and sample-based variant), \emph{modular} (plug and play normal-form selection rule into the actor), and \emph{interpretable} (the long-run selected policies inherit the normal-form selection properties.
These results provide a principled, plug-and-play recipe to steer MARL toward high-performance equilibria, complementing prior work that focused on convergence to an arbitrary NE or on worst-case PoA bounds~\citep{zhang2021gradient,leonardos2021global,zhang2023PoA,chen2022convergence}.

%% file: Sections/prelim.tex
\section{Problem Settings and Preliminaries}
\subsection{The stochastic game model}
We consider the following SG model~$\cM:= \{\cS,\{\cA\}_{i=1}^n,P,r,\rho,H\}$, where there are~$n$ agents, and~$H$ is the horizon of the Markov game.  
At stage~$h$ the state is~$s_{h}\in\cS$ and agent~$i$ chooses action~$a_{i,h}\in\cA_i$. 
We denote the global action as the concatenation of local actions, i.e.,~$a_h = \{a_{1,h}\dots,a_{n,h}\}$, and thus~$\cA = \cA_1\times\cdots\times\cA_n$. 
State transition probabilities are denoted by~$P$, where next state~$s_{i,h+1}$ is determined by the current state~$s_h$ and action~$a_h$, i.e.,~$s_{h+1}\sim P_{h}(\cdot|s_{h}, a_{h})$. At horizon $h$, the stage reward is $r_{h}:\cS\times\cA\to[0,1]$. The initial local distribution is~$\rho$. 
A stochastic policy~~$\pi = \{\pi_h: \cS\rightarrow\Delta(\cA)\}_{h=1}^H$ (where $\Delta(\cA)$ is the probability simplex over $\cA$) specifies a strategy in which agents choose their actions based on the current state in a stochastic fashion. 
Throughout this paper, we use~$[n]$ to denote the set $\{1,2,...n\}$. 
We also use the notation~$-i$ to denote the set of all agents except for agent $i$, i.e.,~$[n]\backslash i$. 
For two integers~$h,h'$ such that $h\le h'$, the notation $h:h'$ is used to denote the set~$\{h, h+1, \dots, h'\}$ .

Given a policy $\pi$, we define the value function as well as $Q$ function as:
 \begin{align*}
  \textstyle V_{i,h}^\pi(s)&:=\bE^\pi \left[\sum_{h'=h}^{H} r_{i,h'}(s_{h'},a_{h'})\Big|s_h = s\right], \\
Q_{i,h}^\pi(s,a)&:=\bE^\pi \left[\sum_{h'=h}^{H} r_{i,h'}(s_{h'},a_{h'})\Big|s_h = s, a_h = a\right].
 \end{align*}
We would also like to note that the value of $V_{i,h}^\pi$ only depends on $\pi_{h:H}$ and that $Q_{i,h}^\pi$ depends on $\pi_{h+1:H}$. Thus, in circumstances where it needs more accurate specifications, we also write $V_{i,h}^\pi$ as $V_{i,h}^{\pi_{h:H}}$ and $Q_{i,h}^\pi$ as $Q_{i,h}^{\pi_{h+1:H}}$.


Given the definition of $Q$-functions, we are now ready to define the Markov perfect NE.


\begin{defi}[Markov perfect (Nash) equilibrium (MPE)]
A policy $\pi^\star$ is called a Markov perfect equilibrium (MPE) if and only if for all $i\in[n], h\in[H],s\in \cS$, the following inequality holds
\vspace{-5pt}
\begin{align*}
\textstyle Q_{i,h}^\pi(s, a^\star) \ge Q_{i,h}^\pi(s, a_i,a_{-i}^\star), ~~\forall a_i \neq a_i^\star,
 \end{align*}
where $a^\star$ is such that $\pi^\star(a^\star|s) > 0$. We also denote the set of MPEs as $\Pi_{\textup{MPE}}$.
\end{defi}
MPEs characterize individually rational behavior in stochastic games.  
Among these equilibria, however, different MPEs may lead to very different collective outcomes.  
This motivates the following notion of social efficiency.
\begin{defi}[Pareto optimal policy and Pareto optimal MPE]\label{defi:pareto-optimal}
A policy $\pi^\star$ is a Pareto optimal policy if it maximizes the social utility, i.e., 
$$\textstyle \sum_{i=1}^n V_{i,h}^{\pi^\star}(s) \ge \sum_{i=1}^n V_{i,h}^{\pi}(s)~\forall~ \pi, h\in[H], s\in\cS.$$

Similarly, a policy $\pi^\star$ said to be a Pareto optimal MPE if $$\textstyle\sum_{i=1}^n V_{i,h}^{\pi^\star}(s) \ge \sum_{i=1}^n V_{i,h}^{\pi}(s)~\forall~ \pi\in\Pi_{\textup{MPE}}.$$
\end{defi}

To analyze equilibrium selection beyond Pareto efficiency, we next introduce a structural class of stochastic games that admits a potential representation, which also plays a central role in our results.

%
\begin{defi}[Markov potential game (MPG)]\label{defi:MPG}
    A stochastic game is called a Markov potential game (MPG) if there exists functions $\{\phi_h\}_{h\in[H]},\phi_h:\cS\times\cA \to \bR$ such that the following function
    \begin{align*}
     \textstyle   \Phi_h^\pi(s,a):= \bE^\pi \left[\sum_{h'=h}^{H} \phi_{h'}(s_{h'},a_{h'})~\Big|s_h = s, a_h = a\right]
    \end{align*}
    satisfies the equation $\Phi_h^\pi(s,a_i, a_{-i}') - \Phi_h^\pi(s,a_i, a_{-i}) = Q_{i,h}^\pi(s,a_i, a_{-i}') - Q_{i,h}^\pi(s,a_i, a_{-i})$
    for any policy $\pi$ and agent $i$ as well as action $a$. The function $\Phi_h^\pi$ is called the total potential function and $\phi$ is the stage potential function.
\end{defi}

 Definitions of MPG also exist in the literature for the setting of continuous state and action spaces \cite{Macua18} as well as infinite-horizon discounted case SGs \cite{zhang2021gradient,leonardos2021global}. 
 {Due to the difference in the settings, the definition is also slightly different compared with these works. 
 However, one aspect in common is that identical interest SG always serves as an important special case of MPG.}

Given the definition of MPG, we define the potential maximizing policy as follows.
\begin{defi}[Potential-maximizing policy]\label{defi:potential maximizing policy}
In this paper, the potential policy is only defined for MPGs (Definition \ref{defi:MPG}). For a MPG with total potential function $\Phi_h^\pi$, the potential-maximizing policy $\pi^\star$ is the policy such that $\Phi_h^{\pi^\star}(s,a) \ge \Phi_h^{\pi}(s,a),~\forall~ s\in\cS, a\in\cA$ holds for all policy $\pi$.
\end{defi}

\subsection{Informal statement of the main result - equilibrium selection for SGs}\label{section:informal-statement} 

Our primary goal is to understand \emph{equilibrium selection} in stochastic games (SGs): 
when multiple Markov perfect equilibria (MPEs) exist, which equilibria are selected by decentralized learning dynamics in the presence of small noise or exploration?

The key insight of this paper is that equilibrium selection in SGs can be reduced to equilibrium selection in normal-form games (details introduced in Section \ref{sec:prelim-equlibrium-selection}) by using state–action value functions as intermediate payoffs. 
Specifically, given any joint policy $\pi$, the $Q$-functions 
$\{Q_{i,h}^{\pi}(s,a)\}_{i,h,s,a}$ define, at each state $s$ and stage $h$, a normal-form game in which each agent $i$ chooses an action $a_i$ to maximize $Q_{i,h}^{\pi}(s,a_i,a_{-i})$. 
We leverage this observation to construct a unified \emph{actor–critic} framework (details introduced in Section~\ref{sec:unified-framework}): 
the critic estimates the $Q$-functions induced by the current policy, and the actor updates the policy by applying a learning rule from normal-form games to these $Q$-values.

This reduction allows us to lift classical equilibrium-selection results from normal-form games to stochastic games. 
In particular, if a learning rule selects certain Nash equilibria in the normal-form game defined by rewards $\{r_i\}$, then applying the same rule to the $Q$-functions will select corresponding MPEs in the stochastic game.

To illustrate the idea, consider \emph{log-linear learning}.  
In a normal-form potential game, log-linear learning selects the potential-maximizing Nash equilibria. In our setting, at each state $s$ and stage $h$, agents apply the same log-linear learning rule to the $Q$-values $Q_{i,h}(s,\cdot)$. 
This induces a policy over actions that favors those maximizing a corresponding Markov potential, and hence selects Markov perfect equilibria that are potential-maximizing in the stochastic game.

More generally, replacing log-linear learning with other learning rules from the normal-form game literature, such as the dynamics in \cite{marden_achieving_2012} or \cite{pradelski2012learning}, yields different equilibrium-selection criteria over Markov perfect equilibria, including Pareto-optimality. 
Table~\ref{table:result-summary} summarizes the resulting correspondences between equilibrium selection in normal-form games and in stochastic games.

\begin{table}[tb]
    \centering
    \caption{Equilibrium selection}
    \scriptsize
\begin{tabular}{|c|c|c|}
\hline 
 \textbf{Learning Rule} & \textbf{For Normal-form game} & \textbf{For Stochastic Game} \\
\hline 
 {log-linear learning } & Potential maximizing (Potential game) & Potential maximizing (MPG)  \\
\hline 
 \cite{pradelski2012learning} & Pareto optimal Pure NE & Pareto optimal MPE \\
 \hline
 \cite{marden_achieving_2012}& Pareto optimal & Pareto optimal  \\
 \hline
\end{tabular}
\label{table:result-summary}
\end{table}

\subsection{Equilibrium selection for normal-form games}\label{sec:prelim-equlibrium-selection}
As stated in Section \ref{section:informal-statement}, our main result builds on equilibrium selection in normal-form games, which serves as a fundamental building block for equilibrium selection in stochastic games. This section formalizes the normal-form game component. The normal-form game is captured by the reward functions $\{r_i: \cA\to \bR\}_{i=1}^n$, where $r_i$ is the reward of agent $i$ and $\cA = \cA_1\times\cdots\times \cA_n$ is the action space. For equilibrium selection, people generally consider the setting where the group of agents follows certain \textit{iterative} learning rules that describe how agents respond to the reward outcome from the previous action. To assist the learning, the iterative process sometimes includes some other auxiliary variables.  Mathematically,  a learning rule can be in general described  by a Markov chain:
\begin{align*} (a\^{t+1}, \xi\^{t+1})\sim \Peps(\cdot,\cdot|a\^{t}, \xi\^{t}),
\end{align*}
where the transition kernel $\Peps$ is defined on the action space $\cA$ and an auxiliary variable space $\cE$, and $a\^{t}\in\cA, \xi\^{t}\in\cE$ stands for the action and auxiliary variables at iteration $t$, respectively. Here the parameter $\epsilon$ in the transition kernel represents the `rate of mistakes', where we assume that agents do not respond fully rational and accurate towards their observation of the reward outcome. We will use the following example, log-linear learning, which is one of the most popular algorithms analyzed in equilibrium selection, to further illustrate the role of $\epsilon$.
\begin{example}[Log-linear learning \citep{blume_statistical_1993,marden_revisiting_2012}]\label{example:log-linear} 
For log-linear learning, there is no hidden variable, i.e., $\cE = \emptyset$, and the algorithm is defined by the following transition probability kernel $\Peps$
\begin{align}
\textstyle 
\Peps\left(a\^{t+1} = (a_{-i}\^{t}, a_i\^{t+1})|a\^{t}; \brac{r_i}\right) = \frac{1}{n}\frac{\epsilon^{- r_i(a\^{t+1})}}{\sum_{ a_i} \epsilon^{-r_i( a_i, a_{-i}\^{t})}}.\label{eq:log-linear}
\end{align}
\end{example}
\vspace{-5pt}
In Example \ref{example:log-linear}, as the rate of mistakes $\epsilon$ goes to zero, the learning rule will converge to the \textit{best-response} strategy. Note that although the Markov chain induced by the best-response strategy is not ergodic and can have multiple stationary distributions (e.g. all strict NEs are stationary), by adding the rate of mistakes $\epsilon$, the Markov chain is ergodic and has a unique stationary distribution. {We also provide another example for better illustration of the hidden variable $\xi$.}
\begin{example}[\citep{marden_achieving_2012}] \label{example:pareto-optimal}
For the learning rule in \citep{marden_achieving_2012}, the hidden variable $\xi$ takes the form $\xi = (\xi_1, \xi_2, \dots,\xi_n)$, where $\xi_i \in \{\C,\D\}$ represents the `mood' of agent $i$, with $\C$ represents Content and $\D$ represents Discontent. Given $a_i\^{t}, \xi_i\^{t}$, the update for $a_i\^{t+1}$ is given by:
\begin{itemize}
    \item If $\xi_i\^{t} = \D$, then $a_i\^{t+1}$ is uniform randomly selected from $\cA_i$.
    \item If $\xi_i\^{t} = \C$, then with probability $p = 1-\epsilon^c$, $a_i\^{t+1} = a_i\^{t}$, with probability $\epsilon^c$, $a_i\^{t+1}$ is uniform randomly selected from $\cA_i\backslash \{a_i\^{t}\}$,
\end{itemize}
where the parameter $c$ is a constant that $c \ge n$. Given $a\^{t+1}$, the update for $\xi_i\^{t+1}$ is given by
\begin{itemize}
    \item If $\xi_i\^{t} = \C$ and $a\^{t+1} = a\^{t}$, then $\xi\^{t+1} = \C$.
    \item Otherwise, $\xi_i\^{t+1} = \C$ with probability $\epsilon^{1-r_i(a\^{t+1})}$ (else, $\xi_i\^{t+1} = \D$).
\end{itemize}

\end{example}

In general, in this paper, we make the following ergodicity assumptions for $\Peps$,
\begin{assump}[Ergodicity]\label{assump: Ergodicity}
    For any given $\epsilon > 0$, and any set of reward $\{r_i: \cA\to \bR\}_{i=1}^n$, the Markov chain induced by the transition kernel $\Peps(\cdot|\cdot; \brac{r_i})$ is ergodic. 
\end{assump}
\begin{defi}[Stationary distribution $\pieps$]
    Given Assumption \ref{assump: Ergodicity}, we know that the Markov chain $\Peps$ induces a unique stationary distribution on $\cA\times\cE$, which we denote as $\pieps(\cdot,\cdot)$  Without causing confusion, we also overload the notation $\pieps(a)$ to denote the marginal stationary distribution on $\cA$, i.e. $\pieps(a) = \sum_{\xi\in\cE}\pieps(a,\xi)$.
\end{defi}

We are now ready to discuss equilibrium selection in detail.  One idea of equilibrium selection is that although there are multiple NEs in a normal-form game, some equilibria are more stable than others when the learning process has mistakes. These more stable NEs are known as stochastically stable equilibria (SSE)~\cite{foster1990stochastic, young_evolution_1993}.  For instance, although all strict NEs are stationary under the best response, interestingly, as the rate of mistakes $\epsilon$ in the log-linear learning rule~\eqref{eq:log-linear} goes to zero, $\pieps$ will only have support on a limited subset of the NEs, which are the SSE.
\begin{defi}[Stochastically stable equilibrium (SSE) \citep{foster1990stochastic, young_evolution_1993}] 
    The stochastically stable equilibria of a learning rule $\Peps$ are actions $a^\star$ such that $\lim_{\epsilon\to 0}\pieps(a^\star) > 0$
\end{defi}
To calculate the SSE given a normal-form game $\brac{r_i}$, we first define the resistance and stochastic potential of the normal-form game. We make the following assumption of the learning rule:
\begin{assump}\label{assump:resistence}
    For any pair of $(a,\xi), (a',\xi')$, we have that there exists a constant $R\left((a,\xi)\to (a',\xi')\right)$  and $C_1, C_2 > 0$, such that
    \begin{align*}
        C_1 \epsilon^{R\left((a,\xi)\to (a',\xi')\right)}  < \Peps(a',\xi'|a,\xi;\brac{r_i}) <  C_2 \epsilon^{R\left((a,\xi)\to (a',\xi')\right)}
    \end{align*}
\end{assump}
\begin{defi}[Resistance]
    The constant $R\left((a,\xi)\to (a',\xi')\right)$ is called as the \emph{resistence} of the transition $(a,\xi)\to (a',\xi')$ under transition kernel $\Peps(\cdot|\cdot;\brac{r_i})$.

    Additionally, we can construct a graph given the resistance of transitions, where the weight of the directed edge $(a,\xi)\to (a',\xi')$ is given by $R((a,\xi)\to (a',\xi'))$. Given a set of edges $T$ in the directed graph, the resistances of these edges are given by the sum of the resistance of edges, i.e. ~$R(T):= \sum_{a\to a' \in T}R(a\to a').$
\end{defi}
Note that when $\lim_{\epsilon\to 0}\Peps(a,\xi'|a,\xi;\brac{r_i}) > 0$, then the resistance $R((a,\xi)\to (a',\xi')) = 0$. If $\Peps(a,\xi'|a,\xi;\brac{r_i}) = 0$ for all $\epsilon$, then the resistance $R((a,\xi)\to (a',\xi')) = + \infty$. 

Given the definition of resistance, we are ready to define the stochastic potential. We first define the \emph{spanning tree} rooted at vertex $(a,\xi)$ as follows. A spanning tree, rooted at vertex $(a,\xi)$ (or a $(a,\xi)$-tree) is a set of $|\cA||\cE|-1$ directed edges such that from every vertex different from $(a,\xi)$, there is a unique directed path in the tree to $(a,\xi)$. We denote the set of $(a,\xi)$-trees as $\cT(a,\xi)$.
\begin{defi}[Stochastic potential]\label{defi:stochastic-potential}
    The stochastic potential $\gamma(a,\xi)$ of an action $a$ and a hidden variable $\xi$ is defined as:
    \begin{equation*}
      \textstyle  \gamma(a,\xi):= \min_{T\in \cT(a,\xi)} R(T)
    \end{equation*}
We also use the notation $\gamma(a,\xi;\brac{r_i})$ to specify that the stochastic potential is with respect to the game $\brac{r_i}$.
\end{defi}
The following theorem points out that the SSEs minimize the stochastic potential
\begin{theorem}[\citep{foster1990stochastic,young_evolution_1993}]\label{theorem:young-evolution}
    An action $a^\star\in\cA$ is the SSE of a normal-form game $\brac{r_i}$ and learning rule $\Peps$ if and only if there exists a hidden variable $\xi^\star\in\cE$ such that $a^\star, \xi^\star$ minimizes the stochastic potential, i.e., $\gamma(a^\star, \xi^\star) = \min_{a,\xi} \gamma(a,\xi)$. Further, for any pair of $(a,\xi), (a',\xi')$,there exists a constant $C>0$ such that
    \vspace{-5pt}
    \begin{align*}
        \qquad \frac{\pieps(a,\xi)}{\pieps(a',\xi')}< C\epsilon^{\gamma(a,\xi) - \gamma(a',\xi')}
    \end{align*}
\end{theorem}
Theorem \ref{theorem:young-evolution} is very useful for capturing the SSE of different learning rules. For example, for potential games (defined below), classical results show that the stochastic potential $\gamma$ is closely related to the potential function $\phi$.
\begin{defi}[Potential game]\label{defi:potential-game}
    A normal-form game $\brac{r_i}$ is called a potential game if there exists a potential function $\phi:\cA\to\bR$ such that the following equation holds for any agent $i$.
    \vspace{-5pt}
    \begin{align*}
        \phi(a_i, a_{-i}) - \phi(a_i,a_{-i}) = r_i(a_i, a_{-i}) - r_i(a_i', a_{-i}), ~~\forall a_i, a_{i}'\in\cA_i, a_{-i}\in \cA_{-i}.
    \end{align*}
\end{defi}
Thus we can obtain the following corollary.
\begin{coro}[\citep{blume_statistical_1993,marden_revisiting_2012}]\label{coro:log-linear-learning}
 For potential games, under log-linear learning (Example \ref{example:log-linear}), the SSEs are the potential maximizing actions $a^\star$ such that $a^\star \in \argmax_a \phi(a)$. Further, the stationary distribution $\pieps$ satisfies
    \vspace{-5pt}
\begin{align*}
  \textstyle  \frac{\pieps(a)}{\pieps(a')} < C\epsilon^{\phi(a') - \phi(a)}.
\end{align*}
\end{coro}
\vspace{-5pt}
Further, for the learning rule in Example \ref{example:pareto-optimal}, we can show that the SSE is pareto-optimal.
\begin{coro}[\cite{marden_achieving_2012}]\label{coro:pareto-optimal}
For general-sum normal-form games, under certain interdependence assumption (cf. \cite{marden_achieving_2012}), the SSEs of the learning rule in Example \ref{example:pareto-optimal} are Pareto-optimal (which are not necessarily NEs), i.e., $a^\star$ is a SSE if and only if
\begin{align*}
  \textstyle \quad  a^\star \in \argmax_{a} \sum_{i=1}^n r_i(a)
\end{align*}
\end{coro}
Here we also provide another learning rule whose SSE is the Pareto-optimal NE. For the sake of compactness, we refer the readers to the original paper for the detailed learning rule. 
\begin{coroexmp}[\cite{pradelski2012learning}]\label{example:pareto-NE} If a general-sum game has at least one pure NE, then under the same interdependence assumption as in Corollary \ref{coro:SG-pareto-optimal}, the learning rule in \citep{pradelski2012learning} guarantees that 
SSEs are the Nash equilibria that are Pareto-optimal, i.e., $a^\star$ is a SSE if and only if $a^\star \in \argmax_{a \textup{ is a pure NE}} \sum_{i=1}^n r_i(a).$
\end{coroexmp}

%% file: Sections/unified_framework.tex
\section{Equilibrium selection for stochastic games: a unified framework}\label{sec:unified-framework}
Given the promising results of equilibrium selection in normal-form games, a natural next step is to extend this approach to stochastic games. In this paper, we present a unified framework that leverages the learning rules $\Peps$ from the normal-form game setting as foundational components for solving stochastic games. This framework's modularity is particularly appealing, allowing us to derive different outcomes by simply incorporating various learning rules.

The algorithm is outlined in Algorithm \ref{alg:unified-framework}. The core idea is to apply the learning rule $\Peps$ independently at each stage $h$ and state $s$. Here, $a_h^t(s)$ and $\xi_h^t(s)$ denote the action and hidden variable at iteration $t$, stage $h$, and state $s$, respectively. The algorithm comprises two main steps: the actor (step 1) and the critic (step 2). The critic updates the variables $Q_{i,h}^t(s,a)$ using a Bellman-like iteration, while the actor applies the learning rule $\Peps$, substituting the reward $\brac{r_i}$ in the normal-form game with the values $\brac{Q_{i,h}^t(s,\cdot)}$ computed by the critic instead. 
\begin{algorithm}
\caption{Unified Learning Framework}\label{alg:unified-framework}
\begin{algorithmic}
\REQUIRE Initialization $Q_{i,h}\^{0}(s,a) = r_{i,h}(s,a)$, $V_{i,H+1}\^{t}(s)=0$, for all $h \in [H], i\in[N], s\in\cS, a\in\cA, t\ge 0$. Randomly initialize $a_h\^{0}(s)\in\cA,\xi_h\^{0}(s)\in\cE$.
\FOR{$t = 0,1,\dots$}
\FOR{$h = H, H-1, \dots, 1$}
\STATE Step 1 (Actor): Let $a_h\^{t+1}(s), \xi_h\^{t+1}(s)\sim \Peps(\cdot,\cdot|a_h\^{t}(s), \xi_h\^{t}(s);\brac{Q_{i,h}\^{t}(s,\cdot)})$ 
\STATE Step 2 (Critic): Calculate $V_{i,h}\^{t+1}, Q_{i,h}\^{t+1}$ for all $i\in[n], h\in[H]$ as follows:
\begin{align*}
    &\textstyle V_{i,h}\^{t+1}(s) \!= \!\frac{t}{t\!+\!1} V_{i,h}\^t (s) + \frac{1}{t\!+\!1} Q_{i,h}\^t(s,a_h\^t(s))\left(\!\textup{i.e.,} V_{i,h}\^{t+1}(s) \!=\! \frac{\sum_{\tau=1}^{t+1} Q_{i,h}\^\tau (s,a_h\^\tau(s))}{t+1}\textup{}\right)\\
    &\textstyle Q_{i,h}\^{t+1}(s,a) = r_{i,h}(s,a) + \sum_{s'}P_h(s'|s,a) V_{i,h+1}\^{t+1}(s')
\end{align*}
\ENDFOR
\ENDFOR
\end{algorithmic}
\end{algorithm}

\begin{defi}[Notation $\pi\^{t}$ and $\pieps$]\label{defi:pieps-SG}
    We use the notation $\pi_h\^{t}(\cdot,\cdot|s)$ to denote the joint probability distribution of $a_h\^{t}(s), \xi_h\^{t}(s)$. We also overload the notation $\pi_h\^{t}$ to also denote the marginal distribution on $a_h\^{t}(s)$, i.e. $\pi_h\^{t}(a|s) = \sum_{\xi\in\cE} \pi_h\^{t}(a,\xi|s)$. If the limitation exists we also denote the stationary distribution as $\pieps_h(a,\xi|s):=\lim_{t\to+\infty} \pi_h\^{t}(a,\xi|s)$, $\pieps_h(a|s):=\lim_{t\to+\infty} \pi_h\^{t}(a|s)$.
\end{defi}
We are now ready to present our main theorem, which analyzes the stochastically stable policy of Algorithm \ref{alg:unified-framework} and extends the equilibrium selection theorem (Theorem \ref{theorem:young-evolution}) to the SG setting.
\begin{theorem}\label{theorem:stochastic-stable-global-optimality}
    Suppose that the learning rule $\Peps$ considered in Algorithm \ref{alg:unified-framework} satisfies Assumption \ref{assump: Ergodicity} and \ref{assump:resistence}, then the stationary distribution $\pieps$ (defined in Definition \ref{defi:pieps-SG}) exists. Further, $\lim_{\epsilon\to 0}\pieps = \pi^*$ exists, and that $\pi_h^*(\cdot|s)$ only has support on actions such that $\exists~\xi\in \cE$ such that $(a,\xi)$ minimizes the potential $\gamma(a,\xi;\brac{Q_{i,h}^{\pi_{h+1:H}^*}(s,\cdot)})$ (Definition \ref{defi:stochastic-potential}), i.e. 
    \begin{align*}
    \textstyle \gamma(a,\xi;\brac{Q_{i,h}^{\pi_{h+1:H}^*}(s,\cdot)}) = \min_{a',\xi'}\gamma(a',\xi';\brac{Q_{i,h}^{\pi_{h+1:H}^*}(s,\cdot)}).
    \end{align*}
\end{theorem}

Theorem \ref{theorem:stochastic-stable-global-optimality} may initially appear abstract and challenging to interpret. To give a better illustration of the theorem, we present the following corollaries that apply Theorem \ref{theorem:stochastic-stable-global-optimality} to different learning rules (Example \ref{example:log-linear}, \ref{example:pareto-optimal} and \ref{example:pareto-NE}). A detailed proof sketch and explanation of Theorem \ref{theorem:stochastic-stable-global-optimality} and its implications are provided in the next section.

\begin{coro}\label{coro:SG-log-linear-learning}
    For a MPG (Definition \ref{defi:MPG}), if the learning rule $\Peps$ in Algorithm \ref{alg:unified-framework} is chosen as log-linear learning (Example \ref{example:log-linear}), then $\pi^\star$ defined in Theorem \ref{theorem:stochastic-stable-global-optimality} is a potential-maximizing policy (Definition \ref{defi:potential maximizing policy}).
\end{coro}

\begin{coro}\label{coro:SG-pareto-optimal}
    For a general sum SG, if the learning rule $\Peps$ in Algorithm \ref{alg:unified-framework} is chosen as Example \ref{example:pareto-optimal} and under certain interdependence assumption (see Assumption \ref{assump:SG-interdependence} in Appendix \ref{apdx:proof-corollaries}), $\pi^\star$ defined in Theorem \ref{theorem:stochastic-stable-global-optimality} is a Pareto optimal policy (Definition \ref{defi:pareto-optimal}).
\end{coro}

\begin{coro}\label{coro:SG-pareto-optimal-NE}
    For a general sum SG, if the learning rule $\Peps$ in Algorithm \ref{alg:unified-framework} is chosen as Example \ref{example:pareto-NE}, and if there exists at least one MPE, then under the interdependence assumption (see Assumption \ref{assump:SG-interdependence} in Appendix \ref{apdx:proof-corollaries}), $\pi^\star$ defined in Theorem \ref{theorem:stochastic-stable-global-optimality} is a Pareto optimal MPE (Definition \ref{defi:pareto-optimal}).
\end{coro}

\begin{rmk}\label{rmk:limitation}
Corollary \ref{coro:SG-log-linear-learning},\ref{coro:SG-pareto-optimal},\ref{coro:SG-pareto-optimal-NE} produce similar stochastically stable outcomes compared with their counterparts in the normal form setting (Corollary \ref{coro:log-linear-learning}, \ref{coro:pareto-optimal}, \ref{example:pareto-NE}), demonstrating the modularity and convenience of our framework. However, it also highlights a limitation: the strength of our results is contingent upon the strength of the corresponding normal-form game results. If the results for the normal-form game rely on certain assumptions (e.g., the potential game assumption or the interdependence assumptions), the corresponding results for the stochastic game will hold only under the same assumptions. Furthermore, since our algorithm applies learning rules at each stage $h$ and state $s$, the assumptions must hold for each stage and state accordingly.
\end{rmk}


\paragraph{Fully sample-based algorithm}
Note that Algorithm 1 is still restrictive because the update of the critic (step 2) still requires knowledge of the true transition probability $P_h(s'|s,a)$. However, we would like to point out that our algorithm framework as well as theoretical analysis can be extended to fully sample-based learning. For the sake of compactness, here we only provide the sample-based algorithm in Algorithm \ref{alg:sample-based unified-framework} but defer the theory and proof into Appendix \ref{apdx:sample-based}.

\begin{rmk}[Comparison with Existing Decentralized MARL Algorithms]
Algorithm~\ref{alg:unified-framework} differs from classical decentralized MARL methods such as \cite{zhang2021gradient,song2021}. In those works, agents do not observe other agents’ actions and maintain local critic estimates that effectively correspond to averaged action-value functions depending only on $(s,a_i)$. In contrast, our update rule requires observing other agents’ actions and maintaining a joint action-value function $Q(s,a)$ over the full action space, which leads to increased storage and sample complexity. This additional complexity, however, enables stronger guarantees. Existing decentralized MARL algorithms typically establish convergence to a Nash equilibrium or a (coarse) correlated equilibrium, without specifying which equilibrium is selected when multiple equilibria exist. By contrast, our framework provides explicit equilibrium selection guarantees, allowing convergence to a designated equilibrium. From this perspective, the proposed algorithm trades increased complexity for finer control over equilibrium selection.
\end{rmk}

\begin{algorithm}
\caption{Sample-based Unified Learning Framework}\label{alg:sample-based unified-framework}
\begin{algorithmic}
\REQUIRE Initialization $Q_{i,h}\^{0}(s,a) = r_{i,h}(s,a)$, $V_{i,H+1}\^{t}(s)=0$, for all $h \in [H], i\in[N], s\in\cS, a\in\cA, t\ge 0$. Randomly initialize $a_h\^{0}(s)\in\cA,\xi_h\^{0}(s)\in\cE$. Set the initial count $N_h(s,a) = 0$ for all $h\in[H], s\in\cS, a\in\cA$.
\FOR{$t = 0,1,\dots$}
\FOR{$h = H, H-1, \dots, 1$}
\STATE Step 1 (Actor): Let $a_h\^{t+1}(s), \xi_h\^{t+1}(s)\sim \Peps(\cdot|a_h\^{t}(s), \xi_h\^{t}(s);\brac{Q_{i,h}\^{t}(s,\cdot)})$ 
\STATE Step 2 (Sample): Randomly sample $\bars_1\^{t+1} \in \cS$, then sample the trajectory $\bars_1\^{t+1},\bara_1\^{t+1}, \bars_2\^{t+1},\dots,\bara_{H-1}\^{t+1}, \bars_H\^{t+1}$ 
using the policy $\bara_h\^{t+1}:= a_h\^{t+1}(\bars_h\^{t+1})$.\\
Update $N_h(\bars_h\^{t+1}, \bara_h\^{t+1}) \leftarrow N_h(\bars_h\^{t+1}, \bara_h\^{t+1}) +1$
\STATE Step 3 (Critic): Calculate $V_{i,h}\^{t+1}, Q_{i,h}\^{t+1}$ for all $i\in[n], h\in[H]$ as follows:
\begin{align*}
    &\hspace{-20pt}\textstyle V_{i,h}\^{t+1}(s) = \frac{t}{t\!+\!1} V_{i,h}\^t (s) + \frac{1}{t\!+\!1} Q_{i,h}\^t(s,a_h\^t(s))\quad \left(\textup{i.e. } V_{i,h}\^{t+1}(s) = \frac{\sum_{\tau=1}^{t+1} Q_{i,h}\^\tau (s,a_h\^\tau(s))}{t+1}\textup{}\right)\\
    &\hspace{-20pt}\textstyle Q_{i,h}\^{t+1}(s,a) \!=\! \left\{\begin{array}{l}
     \hspace{-5pt}{\scalebox{0.87}{$ \frac{N_h(s,a)-1}{N_h(s,a)}Q_{i,h}\^{t}(s,a) \!+\! \frac{1}{N_h(s,a)}\left(r_{i,h}(s,a) \!+\!  V_{i,h+1}\^{t+1}(\bars_{h+1}\^{t+1})\right)\!, \textup{if } s, \!a\!=\! \bars_h\^{t\!+\!1}\!\!, \bara_h\^{t\!+\!1}$}}\\
     Q_{i,h}\^{t}(s,a)   \quad  \textup{otherwise}\qquad 
    \end{array}\right.
\end{align*}
\ENDFOR
\ENDFOR
\end{algorithmic}
\end{algorithm}

%% file: Sections/proof_sketches.tex
\section{Proof sketches for Theorem~\ref{theorem:stochastic-stable-global-optimality}}
\paragraph{Key technical challenge: iteration-varying $\brac{Q_{i,h}\^{t}}$} The major technical difficulty in extending the equilibrium selection result from the normal-form game setting to the stochastic game setting is that the learning rule is applied to $\brac{Q_{i,h}\^{t}(s,\cdot)}$ instead of $\brac{r_{i}}$. Notice that $Q_{i,h}\^{t}(s,\cdot)$ is a random variable that varies with each iteration $t$, and it is potentially correlated with outputs in the past iterations. Thus the dynamics is more complicated to analyze and we cannot directly apply existing results in the normal-form game setting to solve the problem. In this section, we use the special case of identical interest stochastic game with $H=2$ and log-linear learning to illustrate the key proof ideas. The full proof and extension to more general settings can be found in the appendix.

\paragraph{To gain intuition: consider two-stage game $H=2$}
To gain an intuition on the convergence property for Algorithm \ref{alg:unified-framework}, we first consider the stochastic game with only two stages, i.e., $H=2$. For the sake of illustration, we choose the learning rule to be log-linear learning (Example \ref{example:log-linear}) and the rewards to be identical rewards, i.e. $r_{i,h} = r_h$ for every agent $i$ at every stage $h=1,2$. Thus, the Q-function estimations $Q_{i,h}\^{t}$ are also the same and we denote it as $Q_h\^{t}$. Also, note that log-linear learning doesn't require hidden variables ($\cE = \emptyset$) thus the learning algorithm at stage $h$ and state $s$ is given by $a_h\^{t+1}(s)\sim\Peps(\cdot|a_h\^{t}(s); Q_h\^{t}(s,\cdot))$.

At the last stage $h=H= 2$ and at each state $s$, we have that $Q_{H}\^{t}(s,\cdot) = r_H(s,\cdot)$, i.e. $Q_H\^{t}$ does not vary with iteration $t$, thus the dynamic is the same as applying log-linear learning on a normal-form game, with the reward matrix given as $r_H(s,\cdot)$. Hence we can apply the results in the normal-form game setting and conclude that the stationary distribution $\pieps_H(\cdot|s)$ exists for $a_H^{(t)}(s)$. 
Additionally, we can apply the concentration lemmas of the Markov chain at the second stage to conclude that
\vspace{-3pt}
\begin{align*}
 \textstyle    V_H\^{t+1}(s) = \frac{\sum_{\tau=1}^{t+1} r_H (s,a_h\^\tau(s))}{t+1} \asto V^{{\pieps_H}}_H(s).
\end{align*}
Thus this gives that for $h = H-1 = 1$, we have
\begin{equation}\label{eq:Q-t-asto-Q-pi}
\begin{split}
   \textstyle  Q_{h=1}\^{t}(s,a) = &\textstyle r_{h=1}(s,a) + \sum_{s'}P_{h=1}(s'|s,a) V_{H}\^{t}(s') \\
   \asto~ &\textstyle r_{h=1}(s,a) + \sum_{s'}P_{h=1}(s'|s,a) V_{H}^{\pieps_H}(s') = Q_{h=1}^{\pieps_H}(s,a).
\end{split}
\end{equation}
At stage $h=1$, as the previous paragraph on the key technical challenge has pointed out, the problem becomes a bit more complicated because $Q_{{h=1}}\^{t}(s,\cdot)$ is a random variable that is iteration-varying. Thus we cannot directly apply the result from the normal-form game setting. However, we can use the useful observation in equation \eqref{eq:Q-t-asto-Q-pi} that $Q_{h=1}\^{t}(s,a) \asto Q_{h=1}^{\pieps_H}(s,a)$, i.e. although $Q_{h}\^{t}$ are random variables that vary w.r.t. $t$, as $t$ goes to infinity it converges almost surely to a fixed value $Q_{h}^{\pieps_H}$. The following lemma serves as a fundamental lemma in the proof.
\begin{lemma}\label{lemma:auxiliary-main-informal}(Informal, formal statement see Lemma \ref{lemma:auxiliary-main} in Appendix \ref{apdx:auxiliaries})
    For random variables $Q_{h=1}\^{t}(s,\cdot)$ that almost surely converges to $Q_{h=1}^{\pieps_H}(s,\cdot)$, we have that the stationary distribution and the concentration properties of the following random processes $\{a_{h=1}\^{t}(s)\}, \{{a'}_{h=1}\^{t}(s)\}$ in Eq \eqref{eq:algorithm-block} and \eqref{eq:auxiliary-block} are the same.
    \vspace{-5pt}
\begin{align}
  \textstyle a_{h=1}\^{t+1}(s) \sim\Peps(\cdot|a_{h=1}\^{t}(s); Q_{h=1}\^{t}(s,\cdot)) \label{eq:algorithm-block}\\
 \textstyle  {a'}_{h=1}\^{t+1}(s) \sim\Peps(\cdot|{a'}_{h=1}\^{t}(s); Q_{h=1}^{\pieps_H}(s,\cdot)) \label{eq:auxiliary-block}
\end{align}
\end{lemma}
 The above lemma suggests that when considering the asymptotic behavior (e.g. stationary distribution), the original random process $\Peps(\cdot|\cdot; Q_{h=1}\^{t}(s,\cdot))$ is equivalent to considering the auxiliary process $\Peps(\cdot|\cdot; Q_{h=1}^{\pieps_H}(s,\cdot))$. Note that for auxiliary process, $Q_{h=1}^{\pieps_H}(s,\cdot)$ is no longer iteration varying. Thus, we can apply classical results in equilibrium selection (in specific Corollary \ref{coro:log-linear-learning}) and obtain the following lemma.
\begin{lemma}\label{lemma:2-stage-lemma2}
     For a given state $s$ and horizon $h$, define $a^*$ such that $Q_h^{\pieps}(s,a^*) \in \argmax_{a}Q_h^{\pieps}(s,a)$. There exists a uniform constant $C$ (with respect to $\epsilon$), such that for any $\epsilon$, 
    \begin{align}
    \textstyle    \pieps_{h=1}(a|s) &\textstyle < C e^{-\frac{1}{\epsilon}(Q^{\pieps_{H}}_h(s,a^*)-Q^{\pieps_{H}}_h(s,a))}, ~~\forall ~a\in \cA\label{eq:2-stage-1}\\
     \textstyle   \pieps_{h=2}(a|s) &<\textstyle C e^{-\frac{1}{\epsilon}(r_H(s,a^*)-r_H(s,a))}, ~~\forall ~a\in \cA \label{eq:2-stage-2}
    \end{align}
\end{lemma}

From Lemma 2, let $\epsilon \to 0$, then from \eqref{eq:2-stage-2} we can get that for $h=H=2$, $\pieps_h \to \pi_h^\star$, where $\pi_h^\star$ is the optimal policy that maximize the reward, i.e. $\pi_h^\star(\cdot|s) \in \argmax_{\pi_h(\cdot|s)} \sum_a\pi_h(a|s) r_h(s,a)$.

As for $h=1$, note that from the previous argument, we have that as $\epsilon\to 0$, $Q^{\pieps_{H}}_h(s,a) \to Q^{\pi^\star_H}_h(s,a)$, thus we can further conclude that  $\pieps_h\to \pi_h^\star$ where  $\pi_h^\star(\cdot|s) \in \argmax_{\pi_h(\cdot|s)} \sum_a\pi_h(a|s)  Q^{\pi^\star_H}_h(s,a)$. Note that the above definitions for $\pi_h^\star$ satisfy the Bellman optimality condition, which allows us to draw the conclusion that $\pi^\star$ is the optimal policy that maximizes the accumulative reward. Thus we have proved that, for an \emph{identical interest, 2-stage game with the log-linear learning algorithm}, $\pieps$ (Definition \ref{defi:pieps-SG}) exists. Further, as $\epsilon\to 0$, $\pieps \to \pi^*$, where $\pi^*$ is the global optimal policy. 

\paragraph{Extension to $H > 2$, general-sum games and other learning rules} For the sake of compactness, we defer the extension to general settings to the appendix. The proof idea follows the above discussion for $H=2$, where the key insight is to use Lemma \ref{lemma:auxiliary-main-informal} to simplify the analysis.

%% file: Sections/numerics.tex
\section{Numerical case study: a two-stage stag-hunt}\label{sec:numerics}

\begin{figure}[htbp]
\centering
\begin{small}
\input{tikfigures/2x2stag_hunt}
\end{small}
\caption{Game schematic for the two-stage stag-hunt game (states, transitions, payoffs).}
\label{fig:stag-hunt-setup}
\end{figure}

We illustrate our framework on a simple two-player, two-stage stag-hunt stochastic game that highlights the equilibrium-selection trade-off.
Each player has two actions: ``Stag'' (action 0) and ``Hare'' (action 1). 
Hunting a Stag requires both players to coordinate and persist in choosing Stag for both stages; if successful the game yields a large joint reward 7.5 allocated evenly (each player receives 3.75) at the second stage. 
Hunting a Hare is easier, one player can catch a Hare in a single stage, but the payoff is smaller (total reward 2, split equally when both choose Hare, so each receives 1). 
The transition structure and payoffs are illustrated in the schematic in Figure~\ref{fig:stag-hunt-setup}.

This game admits two strict MPEs: the \emph{Stag} equilibrium (both players choose Stag at both stages), which is Pareto-dominant, and the \emph{Hare} equilibrium (both choose Hare at both stages), which is risk-dominant~\citep{harsanyi1995new}. The model-based portion of our theory predicts the following selection behavior when we embed well-known normal-form learning rules into our actor (see Section~\ref{sec:unified-framework}): (i) with log-linear learning (soft best responses with vanishing noise), the dynamics select the risk-dominant (Hare) equilibrium in the zero-noise limit; (ii) with the tailored Pareto-selection dynamics of \citet{marden_achieving_2012}, the dynamics instead select the Pareto-dominant (Stag) equilibrium.

For the plots shown in \cref{fig:numerics-apdx} we run the model-based actor–critic variant of Algorithm~\ref{alg:unified-framework} with the critic performing Bellman-style backward updates and the actor at each state/stage implementing either (A) log-linear learning or (B) the Pareto-selection rule of \citet{marden_achieving_2012}. 
To estimate long-run selection frequencies we run~$N=100$ independent trials with different random seeds. 
Each trial consists of~$T$ actor-update iterations (we used~$T=10,000$ in these experiments) and empirical policy frequencies~$\bar\pi^T(a\mid s)$ are computed via time-averaging of the actor's mixed policies. 
The plots display the evolution of~$\bar\pi^t(a\mid s_0)$ for the initial state~$s_0$; the shaded bands show the empirical 60\% confidence band across the~$N$ trials.
Initial policies were uniform; the log-linear actor used a decaying noise schedule so that the effective temperature vanishes slowly; the Pareto-selection rule used the parameter choices recommended in \citet{marden_achieving_2012}.

\cref{fig:numerics-apdx} shows the time evolution of the empirical probability of selecting the two pure joint actions~$a=(1,1)$ (Hare--Hare, yellow) and \(a=(0,0)\) (Stag--Stag, blue) at the initial state. 
With log-linear learning (left plot) the empirical frequency concentrates on the Hare joint action: the yellow curve grows toward one and the blue curve decays toward zero, confirming selection of the risk-dominant equilibrium. 
With the Pareto-selection dynamics (right plot) the opposite behavior occurs: the empirical frequency concentrates on the Stag joint action, corroborating that the tailored dynamics steer play to the Pareto-optimal MPE. 
The shaded 60\% bands indicate the variability across trials and show that the selection outcome is robust under stochastic initialization.

These numerical results align with our theoretical claims: the actor–critic reduction faithfully inherits the stochastic-stability properties of the embedded normal-form rules, allowing us to steer MARL toward qualitatively different equilibria by swapping only the actor dynamics. In Appendix \ref{apdx:another-numerical-example} we also provide additional numerical illustrations for another identical-interest SG example.\footnote{Code: \url{https://github.com/DianYu420376/Equilibrium-Selection-for-MARL}.}


\begin{figure}[tb]
  \centering
\includegraphics[width=0.4\textwidth]{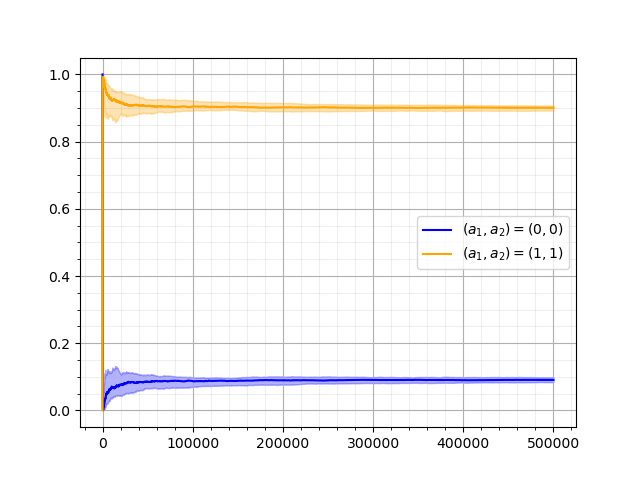}
\hspace{20pt}
\includegraphics[width=0.4\textwidth]{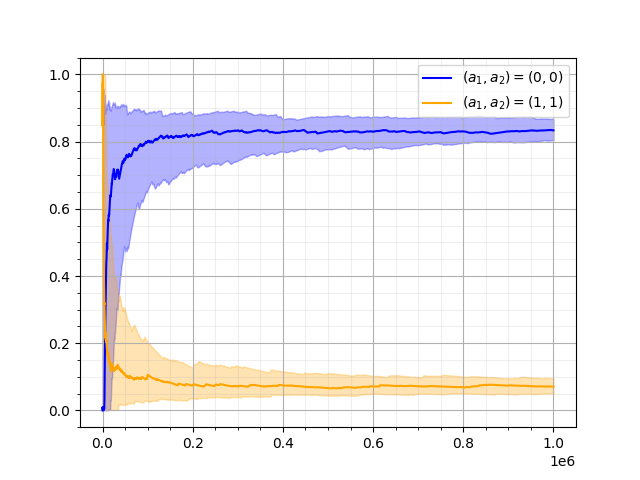}
  \caption{Numerical result of running Algorithm \ref{alg:unified-framework} with log-linear learning (Left) and the learning rule in \citep{marden_achieving_2012} (Right) The plot demonstrates the evolution of $\pi\^{t}(a|s)$ (estimated using empirical averaging), where $s$ is the initial state and the yellow line plots the curve for $a = (1,1)$ and the blue line $a=(0,0)$. The shaded areas are the $60\%$ confidence interval, which are calculated by 100 runs.}
  \label{fig:numerics-apdx}
\end{figure}

%% file: tikfigures/2x2stag_hunt.tex
\begin{center}
{
\vspace{10pt}
\begin{tabular}{|c|c|c|}
\hline 
  & $\displaystyle a_{2} \ =\ 0$ & $\displaystyle a_{2} =1$ \\
\hline 
 $\displaystyle a_{1} =0$ & $\displaystyle 0,0$ & $\displaystyle 0,2$ \\
\hline 
 $\displaystyle a_{1} =1$ & $\displaystyle 2,0$ & $\displaystyle 1,1$ \\
 \hline
\end{tabular}}

\vspace{5pt}
{ $\displaystyle h=1$}
\end{center}

{
\tikzset{every picture/.style={line width=0.75pt}} 
\begin{center}
\hspace{-30pt}
    \begin{tikzpicture}
[x=0.75pt,y=0.75pt,yscale=-0.6,xscale=0.4]

\draw    (353,2.8) -- (137.93,64.22) ;
\draw [shift={(136,64.8)}, rotate = 345.23] [color={rgb, 255:red, 0; green, 0; blue, 0 }  ][line width=0.75]    (10.93,-3.29) .. controls (6.95,-1.4) and (3.31,-0.3) .. (0,0) .. controls (3.31,0.3) and (6.95,1.4) .. (10.93,3.29)   ;
\draw    (353,2.8) -- (557.08,64.22) ;
\draw [shift={(559,64.8)}, rotate = 196.75] [color={rgb, 255:red, 0; green, 0; blue, 0 }  ][line width=0.75]    (10.93,-3.29) .. controls (6.95,-1.4) and (3.31,-0.3) .. (0,0) .. controls (3.31,0.3) and (6.95,1.4) .. (10.93,3.29)   ;

\draw (30,23) node [anchor=north west][inner sep=0.75pt]    {$a_{1} =a_{2} =0$};
\draw (485,23) node [anchor=north west][inner sep=0.75pt]    {Else};
\end{tikzpicture}
\end{center}

}
\begin{center}

{ $\displaystyle  \begin{array}{ccc}
s=A\ \ \ \ \ \ \ \ \ \ \ \ \ \ \ \ \ \ \ \ \ \ \ \ \ \ \ \ \ \ \ \ \ \ \ \ \ \  s =B 
\end{array}$}\\
{
\begin{tabular}{|c|c|c|}
\hline 
  & $\displaystyle a_{2} \ =\ 0$ & $\displaystyle a_{2} =1$ \\
\hline 
 $\displaystyle a_{1} =0$ & $\displaystyle 0.5$ & $\displaystyle 0$ \\
\hline 
 $\displaystyle a_{1} =1$ & $\displaystyle 0$ & $\displaystyle 0$ \\
 \hline
\end{tabular} \ \ \ \ \ \ \ \ 
\begin{tabular}{|c|c|c|}
\hline 
  & $\displaystyle a_{2} \ =\ 0$ & $\displaystyle a_{2} =1$ \\
\hline 
 $\displaystyle a_{1} =0$ & $\displaystyle 0,0$ & $\displaystyle 0,2$ \\
\hline 
 $\displaystyle a_{1} =1$ & $\displaystyle 2,0$ & $\displaystyle 1,1$ \\
 \hline
\end{tabular}}

\vspace{5pt}
{ $\displaystyle h=2$}
\end{center}

%% file: Sections/conclusion.tex
\section{Conclusions}
In this paper, we have presented a unified and modular framework for equilibrium selection in the multi-agent reinforcement learning (MARL) setting. By leveraging classical equilibrium selection results from normal-form games, we extend these insights to the stochastic game context, providing an adaptable approach for achieving higher-quality equilibria. Our framework's modularity allows for the seamless integration of various learning rules, demonstrating its flexibility and wide applicability. Our work admittedly has its limitations. One major drawback is that our results focus primarily on asymptotic guarantees. Important aspects such as convergence rates and sample complexity remain unexplored in our current framework. Additionally, although we managed to obtain asymptotic results for fully sample-based learning, the sample efficiency is not ideal and thus there is still a gap towards application to practical settings. We leave the investigation of these aspects as future work.

\section*{Acknowledgment}

The authors acknowledge the use of AI-assisted tools for improving grammar, clarity, and overall writing quality.

%% file: Sections/appendix_proofs.tex
\section{Proof of Theorem \ref{theorem:stochastic-stable-global-optimality}}\label{apdx:proo-of-main-theorem}
Before stating the proof, we first introduce some important lemmas.
\begin{lemma}\label{lemma:convergence-to-pi-eps} (Proof see Appendix \ref{apdx:convergence-to-pi-eps-optimality})
    For $h = 1,2,\dots, H$, and every $s\in\cS$, the stationary distribution $\pieps$ exists, i.e.  $a_h\^{t}(s)$ converges to $\pieps_h(\cdot|s)$ in distribution, and $V_h\^t(s)\asto V_h^\pieps(s)$, $Q_{i,h}\^t(s,a)\asto Q_{i,h}^\pieps(s,a)$. Further, the stationary distribution can be calculated by the following recursive definition:
    \begin{equation*}
\begin{split}
    &\pieps_H(\cdot,\cdot|s) \textup{ is the stationary distribution under transition kernel }\Peps(\cdot|\cdot;\brac{r_{i\!,H}(s,\!\cdot))}\\
    &\pieps_h(\cdot,\cdot|s)\textup{ is the stationary distribution under } \Peps(\cdot|\cdot;\brac{Q_{i,h}^{\pieps_{h+1:H}}(s,\cdot)}), h\!\le\! H\!-\!1
\end{split}
\end{equation*}
\end{lemma}

\begin{lemma}\label{lemma:pi-eps-optimality}(Proof see Appendix \ref{apdx:convergence-to-pi-eps-optimality})
There exists a uniform constant $C$ (with respect to $\epsilon$), such that for any $\epsilon$, 
    \begin{align*}
        \pieps_h(a,\xi|s) < C\epsilon^{\gamma(a,\xi;\brac{Q^\pieps_{i,h}(s,\cdot)})-\min_{a',\xi'}\gamma(a',\xi';\brac{Q^\pieps_{i,h}(s,\cdot)})}, ~~\forall ~a,\xi\in \cA\times\cE
    \end{align*}
\end{lemma}

We are now ready to state the main proof of this section.
\begin{proof}[Proof of Theorem \ref{theorem:stochastic-stable-global-optimality}]
From Lemma \ref{lemma:convergence-to-pi-eps} we know that $\pi\^{t}(\cdot|s)$ converges to $\pieps(\cdot|s)$ in distribution. Further, from Lemma \ref{lemma:pi-eps-optimality},
\begin{align*}
        \pieps_h(a,\xi|s) < C\epsilon^{\gamma(a,\xi;\brac{Q^\pieps_{i,h}(s,\cdot)})-\min_{a',\xi'}\gamma(a',\xi';\brac{Q^\pieps_{i,h}(s,\cdot)})}, ~~\forall ~a,\xi\in \cA\times\cE
    \end{align*}
    We now prove the theorem using induction. For $h= H$, clearly this implies that $\pieps_H(a,\xi|s)\to 0$ if $\gamma(a,\xi;\brac{r_{i,H}(s,\cdot)})> \min_{a',\xi'}\gamma(a',\xi';\brac{r_{i,H}(s,\cdot)})$, i.e.\\ $\pi_H^*(a|s) = 0$ if for all $\xi\in\cE$, $\gamma(a,\xi;\brac{r_{i,H}(s,\cdot)})> \min_{a',\xi'}\gamma(a',\xi';\brac{r_{i,H}(s,\cdot)})$, which proves the statement for $h=H$.
    
    Assume that the statement holds for $h' = h+1,\dots, H$, $\pieps_{h'}\to \pi_{h'}^*$. We now show that the statement also holds for $h$. Note that for any pair of action and hidden variable $a,\xi$ such that $\gamma(a,\xi;\brac{Q_{i,h}^{\pi_{h+1:H}^*}(s,\cdot)}) >\min_{a',\xi'}\gamma(a',\xi';\brac{Q_{i,h}^{\pi_{h+1:H}^*}(s,\cdot)})$, we have that
    \begin{align*}
       & \lim_{\epsilon\to 0}\pieps_h(a,\xi|s) < \lim_{\epsilon\to 0}C\epsilon^{\gamma(a,\xi;\brac{Q^{\pieps_{h+1:H}}_{i,h}(s,\cdot)})-\min_{a',\xi'}\gamma(a',\xi';\brac{Q^{\pieps_{h+1:H}}(s,\cdot)})}\\
        & = \lim_{\epsilon\to 0}C\epsilon^{\gamma(a,\xi;\brac{Q^{\pi^*_{h+1:H}}{i,h}(s,\cdot)})-\min_{a',\xi'}\gamma(a',\xi';\brac{Q^{\pi^*_{h+1:H}}_{i,h}(s,\cdot)})+ o(1)}  = 0
    \end{align*}
    which completes the proof.
\end{proof}
\section{Proof of Corollary \ref{coro:SG-log-linear-learning}, \ref{coro:SG-pareto-optimal}, \ref{coro:SG-pareto-optimal-NE}}\label{apdx:proof-corollaries}
\begin{proof}[Proof of Corollary \ref{coro:SG-log-linear-learning}]
    Since at $h=H$, for any policy $\pi$, $Q_{i,H}^\pi = r_{i,H}$, thus from Theorem \ref{theorem:stochastic-stable-global-optimality}, we have that at stage $h=H$, $\pi_H^\star(a|s) > 0$ if and only if 
     \begin{align*}
    \textstyle \gamma(a;\brac{r_{i,H}(s,\cdot)}) = \min_{a'}\gamma(a';\brac{r_{i,H}(s,\cdot)(s,\cdot)}).
    \end{align*}
    (Note that for log-linear learning, there's no hidden variable $\xi$.) Under the MPG assumption, we have that $\brac{r_{i,H}(s,\cdot)}$ forms a potential game where the potential function is given as $\Phi_H$. Thus, from Corollary \ref{coro:log-linear-learning} we know that such $a$ satisfies $a \in \argmax_{a} \Phi_H(s,a).$ This implies that at $h=H$, $\pi_H^\star$ is a potential-maximizing policy.

    We now prove by induction. Assume that $\pi_{h+1:H}^\star$ is potential maximizing, we now prove that $\pi_{h:H}^\star$ is also potential maximizing. From the Bellman optimality condition, it suffices to show that $\pi_h^\star(a|s) > 0$ only if $a \in \argmax \Phi_h^{\pi^\star_{h+1:H}}(s,a)$. To show this, once again we apply Theorem \ref{theorem:stochastic-stable-global-optimality} and conclude that $\pi_h^\star(a|s) > 0$ only if 
    \begin{align*}
    \textstyle \gamma(a;\brac{Q_{i,h}^{\pi_{h+1:H}^*}(s,\cdot)}) = \min_{a'}\gamma(a';\brac{Q_{i,h}^{\pi_{h+1:H}^*}(s,\cdot)}),
    \end{align*}
    and Corollary \ref{coro:log-linear-learning} implies that this is equivalent to $a \in \argmax_a \Phi_h^{\pi^\star_{h+1:H}}(s,a),$ and thus the proof is completed by induction
\end{proof}

Before proving Corollary \ref{coro:SG-pareto-optimal} and \ref{coro:SG-pareto-optimal-NE}, we first state the interdependence assumption.
\begin{assump}[Interdependent Stochastic Game]\label{assump:SG-interdependence}
    A stochastic game is called interdependent if for any policy $\pi$, stage $h\in[H]$ and state $s\in\cS$, the Q-functions $\brac{Q_{i,h^\pi(s,\cdot)}}$ forms an interdependent normal-form game (cf. Definition 3.1 in \cite{marden_achieving_2012}).
\end{assump}

\begin{proof}[Proof of Corollary \ref{coro:SG-pareto-optimal}]
   The proof resembles the proof of Corollary \ref{coro:SG-log-linear-learning}. From Theorem \ref{theorem:stochastic-stable-global-optimality}, we have that at stage $h=H$, $\pi_H^\star(a|s) > 0$ if and only if there exists an action $\xi'$ such that $\gamma(a,\xi;\brac{r_{i,H}(s,\cdot)}) = \min_{a',\xi'}\gamma(a',\xi';\brac{r_{i,H}(s,\cdot)}).$
   
    Under Assumption \ref{assump:SG-interdependence}, Corollary \ref{coro:pareto-optimal} further suggests that such $a$ satisfies $a \in \argmax_{a} \sum_{i=1}^n r_{i,H}(s,a)$. This implies that at $h=H$, $\pi_H^\star$ is Pareto-optimal.
We now prove by induction. Assume that $\pi_{h+1:H}^\star$ is Pareto-optimal, we now prove that $\pi_{h:H}^\star$ is also Pareto-optimal. From the Bellman optimality condition, it suffices to show that $\pi_h^\star(a|s) > 0$ only if $a \in \argmax \sum_{i=1}^n Q_{i,h}^{\pi^\star_{h+1:H}}(s,a)$. To show this, once again we apply Theorem \ref{theorem:stochastic-stable-global-optimality} and conclude that $\pi_h^\star(a|s) > 0$ only if 
    \begin{align*}
    \textstyle \gamma(a,\xi;\brac{Q_{i,h}^{\pi_{h+1:H}^*}(s,\cdot)}) = \min_{a',\xi'}\gamma(a',\xi';\brac{Q_{i,h}^{\pi_{h+1:H}^*}(s,\cdot)}).
    \end{align*}
    and Corollary \ref{coro:log-linear-learning} implies that (Under Assumption \ref{assump:SG-interdependence}) this is equivalent to
    \vspace{-5pt}
    \begin{align*}
       \textstyle a \in \argmax_a \sum_{i=1}^n Q_{i,h}^{\pi^\star_{h+1:H}}(s,a),
    \end{align*}
    and thus the proof is completed by induction.
    
\end{proof}

\begin{proof}[Proof of Corollary \ref{coro:SG-pareto-optimal-NE}]
   The proof follows almost the same as the proof of \ref{coro:SG-pareto-optimal}. 

    
\end{proof}

\section{Proof of Lemma \ref{lemma:convergence-to-pi-eps} and \ref{lemma:pi-eps-optimality}}\label{apdx:convergence-to-pi-eps-optimality}
\begin{proof}[Proof of Lemma \ref{lemma:convergence-to-pi-eps}]
    We prove by induction. For $h=H$, the conclusion holds from classical concentration results for Ergodic Markov chains we get that the stationary distribution for $\pi\^t_H(s),\xi\^{t}_H(s)$ is $\pieps_H(\cdot,\cdot|s)$ and that $ V_{i,H}\^{t+1}(s) \!=\! \frac{\sum_{\tau=1}^{t+1} r_{i,H} (s,\pi_h\^\tau(s))}{t+1} \!\asto\! V^{\pieps_H}_{i,H}(s,a)$.
    
    Suppose that the conclusion holds for $h+1$, we now prove for $h$. Since

    $$\textstyle Q_{i,h}\^{t+1}(s.a) = r_{i,h}(s,a) + \sum_{s'}P_h(s'|s,a) V_{i,h+1}\^{t+1}(s').$$
    From from the induction hypothesis $V_{i,h+1}\^{t}\asto V_{i,h+1}^\pieps$ we have that
    \begin{align*}
      \textstyle  Q_{i,h}\^{t+1}(s.a)\asto r_{i,h}(s,a) + \sum_{s'}P_h(s'|s,a) V_{i,h+1}^\pieps(s') = Q_{i,h}^\pieps.
    \end{align*}
    Further, the learning dynamics satisfies 
        \begin{talign*}
         &\Pr\left(a_h\^{t+1}(s), \xi_h\^{t+1}(s), \brac{Q_{i,h}\^{t+1}(s,\cdot)}|a_h\^{1:t}(s), \xi_h\^{1:t}(s),\brac{Q_{i,h+1}\^{1:t+\tau}(s,\cdot)}\right)\\
         &= \Peps(a_h\^{t+1}(s),\xi_h\^{t+1}(s)|a_h\^t(s), \xi_h\^{t}(s); \brac{Q_{i,h}\^t(s)}).
    \end{talign*}
 Additionally, since $Q_{h}\^{t}\asto Q_{h}^\pieps$. Thus, we can apply Lemma \ref{lemma:auxiliary-main} and get that $a_h\^t(s), \xi_h\^{t}(s)$ converges to $\pieps_h(\cdot,\cdot|s)$ (defined as in the statement of Lemma \ref{lemma:convergence-to-pi-eps}) in distribution and that
    \begin{talign*}
        \sum_{\tau=1}^t\frac{Q^\pieps_{i,h}(s,a_h\^{\tau}(s))}{t}\asto \sum_{a}\pieps_h(a|s)Q_{i,h}^\pieps(s,a) = V^\pieps_{i,h}(s).
    \end{talign*}
    Additionally, since  $Q_{i,h}\^{t}\asto Q_{i,h}^\pieps$, we also get $V_{i,h}\^{t}(s) =  \sum_{\tau=1}^t\frac{Q\^{\tau}_{i,h}(s,a_h\^{\tau}(s))}{t}\asto V^\pieps_{i,h}(s),$ which completes the proof.
\end{proof}

\begin{proof}[Proof of Lemma \ref{lemma:pi-eps-optimality}]
We define 
\begin{talign*}
    \mu(a,\xi) = \sum_{T\in \cT(a,\xi)} \prod_{a',\xi'\to a'',\xi'' \in T}\Peps(a'',\xi''|a',\xi';\brac{Q^\pieps_{i,h}(s,\cdot)}).
\end{talign*}
It may be verified that the stationary distribution $\pieps(\cdot,\cdot|s)$ satisfies \citep{freidlin2012random}
\begin{talign*}
    \textstyle\frac{\pieps_h(a,\xi|s)}{\pieps_h(a',\xi|s)} = \frac{\mu(a,\xi)}{\mu(a',\xi')}, ~~\forall~a,\xi,a',\xi'.
\end{talign*}
Further, from Assumption \ref{assump:resistence}, $\mu(a)$ satisfies
\begin{talign*}
    \mu(a,\xi) &= \sum_{T\in \cT(a,\xi)} \prod_{a',\xi'\to a'',\xi'' \in T}\Peps(a'',\xi''|a',\xi';\brac{Q^\pieps_{i,h}(s,\cdot)})\\
    &\le \sum_{T\in \cT(a,\xi)}\prod_{a',\xi'\to a'',\xi'' \in T}C_2\epsilon^{R(a',\xi'\to a'',\xi'';\brac{Q^\pieps_{i,h}(s,\cdot)})}\\
    &\hspace{-25pt}= C_2^{(|\cA||\cE|-1)}\sum_{T\in\cT(a,\xi)}\epsilon^{R(T;\brac{Q^\pieps_{i,h}(s,\cdot)})}\le C_2^{(|\cA||\cE|-1)}|\cT(a,\xi)| {\epsilon}^{\gammaeps(a,\xi;\brac{Q^\pieps_{i,h}(s,\cdot)})}.
\end{talign*}
Similarly
\begin{talign*}
\mu(a,\xi) &\ge \sum_{T\in \cT(a,\xi)}\prod_{a',\xi'\to a'',\xi'' \in T}C_1\epsilon^{R(a',\xi'\to a'',\xi'';\brac{Q^\pieps_{i,h}(s,\cdot)})}\\
    &= C_1^{(|\cA||\cE|-1)}\sum_{T\in\cT(a,\xi)}\epsilon^{R(T;\brac{Q^\pieps_{i,h}(s,\cdot)})}\ge C_1^{(|\cA||\cE|-1)}{\epsilon}^{\gammaeps(a,\xi;\brac{Q^\pieps_{i,h}(s,\cdot)})}.
\end{talign*}
Thus we get
\begin{talign*}
  \!\!  \frac{\pieps_h(a,\xi|s)}{\pieps_h(a^*\!,\xi^*|s)} \!=\! \frac{\mu(a,\xi)}{\mu(a^*\!,\xi^*)} \!\le\!\left(\frac{C_2}{C_1}\right)^{\!(|\cA|-1)} \hspace{-25pt}\max_{a,\xi}|\cT(a,\xi)| \epsilon^{\gammaeps(a,\xi;\brac{Q^\pieps_{i,h}(s,\cdot)})\!-\!\gammaeps(a^*,\xi^*;\brac{Q^\pieps_{i,h}(s,\cdot)})}
\end{talign*}
Then setting $ a^*, \xi^* := \argmin_{a',\xi'}\gammaeps(a',\xi';\brac{Q^\pieps_{i,h}(s,\cdot)})$
completes the proof.
\end{proof}

\section{Fully sampled-based algorithm}\label{apdx:sample-based}
The result for this section relies on the following mild assumptions.
\begin{assump}\label{assump:strong-ergodicity}
    For the learning rule $\Peps$ and any set of (possibly time-varying) reward functions $\brac{r_i\^{t}}$, we have that the (possibly non-homogeneous) Markov chain given by $a\^{t+1},\xi\^{t+1} \sim \Peps(\cdot,\cdot| a\^{t},\xi\^{t}; \brac{r_i\^{t}})$ satisfies that for any $a$, $\Pr(a \textup{ is visited infinitely many times}) = 1.$
\end{assump}

Note that this assumption is a slightly stronger assumption than ergodicity because we require the reward function for the transition kernel to be time-varying. However, it is not a strong assumption as one can easily check that the learning rules in Example \ref{example:log-linear}, \ref{example:pareto-optimal}, and \ref{example:pareto-NE} all satisfy this assumption.
\begin{defi}[Reachable state]
    A state $s$ is reachable at stage $h$ is there exists a sequence of action $s_0, a_1, a_2, \dots, a_{h-1}$ such that there exists a positive probability that state $s$ is being visited at stage $h$, i.e., $\Pr(s_h = s|s_0, a_1, a_2, \dots, a_{h-1}) > 0$
\end{defi}
\begin{assump}\label{assump:reachability}
    All state $s\in \cS$ are reachable at any stage $h\in[H]$.
\end{assump}
We would also like to note that Assumption \ref{assump:reachability} is also a mild assumption. For SGs that doesn't satisfy this assumption, we can always restrict our discussion to the reachable states. 
We are now ready to state the main results for our fully sample-based algorithm (Algorithm \ref{alg:sample-based unified-framework}). The results is almost identical to Theorem \ref{theorem:stochastic-stable-global-optimality}.
\begin{theorem}\label{theorem:sample-based} 
    Suppose that the learning rule $\Peps$ considered in Algorithm \ref{alg:sample-based unified-framework} satisfies Assumption \ref{assump:strong-ergodicity} and \ref{assump:resistence} and that the SG satisfies Assumption \ref{assump:reachability}, then for $h = 1,2,\dots, H$, and every $s\in\cS$, $a\^{t}(s)$ converges to  $\pieps_h(\cdot|s)$ in distribution. Further, $\lim_{\epsilon\to 0}\pieps = \pi^*$ exists, and that $\pi_h^*(\cdot|s)$ only has support on the actions such that $\exists~\xi\in \cE$ such that $(a,\xi)$ minimizes the potential $\gamma(a,\xi;\brac{Q_{i,h}^{\pi_{h+1:H}^*}(s,\cdot)})$, i.e. $$\textstyle \gamma(a,\xi;\brac{Q_{i,h}^{\pi_{h+1:H}^*}(s,\cdot)}) = \min_{a',\xi'}\gamma(a',\xi';\brac{Q_{i,h}^{\pi_{h+1:H}^*}(s,\cdot)}).$$
\end{theorem}
It suffices to prove that Lemma \ref{lemma:convergence-to-pi-eps} still holds for the sample-based algorithm (Algorithm \ref{alg:sample-based unified-framework}). Then the proof of Lemma \ref{lemma:convergence-to-pi-eps} and Theorem \ref{theorem:sample-based} follows the same procedure as the proof of Theorem \ref{theorem:stochastic-stable-global-optimality} (Appendix \ref{apdx:proo-of-main-theorem}). We now prove Lemma \ref{lemma:convergence-to-pi-eps}.

\begin{proof}[Proof of Lemma \ref{lemma:convergence-to-pi-eps} for sample-based version (Algorithm \ref{alg:sample-based unified-framework})]
\!Similar to the proof in Appendix \ref{apdx:convergence-to-pi-eps-optimality} for original version, we still prove by induction. Note that for $h=H$, the conclusion still holds from classical concentration results for Ergodic Markov chains we get that the stationary distribution for $\pi\^t_H(s),\xi\^{t}_H(s)$ is $\pieps_H(\cdot,\cdot|s)$ and that $ V_{i,H}\^{t+1}(s) = \frac{\sum_{\tau=1}^{t+1} r_{i,H} (s,\pi_h\^\tau(s))}{t+1} \asto V^{\pieps_H}_{i,H}(s,a)$.
    
    Suppose that the conclusion holds for $h+1$, we now prove for $h$. Notice that the Q update is the follows:
    \begin{align*}
 \textstyle Q_{i,h}\^{t\!+\!1}\!(s,\!a) \!=\! \left\{\begin{array}{l}
     \hspace{-5pt}\frac{N_h(s,a)-1}{N_h(s,a)}Q_{i,h}\^{t}(s,a) \!+\! \frac{1}{N_h\!(s\!,a)}\!\left(\!r_{i,h}(s,\!a) \!+\!  V_{i,h+1}\^{t\!+\!1}(\bars_{h+1}\^{t\!+\!1})\!\right)\!, \textup{if }s, \!a\!=\! \bars_h\^{\!t\!+\!1\!}\!\!, \bara_h\^{\!t\!+\!1\!}\!\!\!\\
     Q_{i,h}\^{t}(s,a)   \quad  \textup{otherwise}\qquad 
    \end{array}\right.
\end{align*}

Thus it can be concluded that
\begin{talign*}
    Q_{i,h}\^{t\!+\!1}(s,\!a) &\!=\! \frac{1}{N_h^t(s\!,a)} \!\sum_{i\!=\!1}^t\!\! \left(\!r_{i,h}(s,\!a) \!+\! V_{i,h+1}\^{\tau_i}\!(\bars_{h+1}\^{\tau_i}\!)\!\right)\!=\!r_{i,h}(s,\!a) \!+\! \frac{1}{N_h^t(s,\!a)} \!\sum_{i=1}^t\!\!V_{i,h+1}\^{\tau_i}(\bars_{h+1}\^{\tau_i}\!)\!\!\!\!\!
\end{talign*}
where the notation $N_h^t(s,a)$ is the count $N_h(s,a)$ in Algorithm \ref{alg:sample-based unified-framework} at iteration step $t$. The notation $\tau_i$'s are the iteration steps before iteration step $t$ such that $s,a = \bars_h\^{\tau_i},\bara_h\^{\tau_i}$.

Note that given Assumption \ref{assump:strong-ergodicity} and \ref{assump:reachability}, we have that $N_h^t(s,a) \to +\infty$ as $t\to +\infty$. Further, from induction assumption $V_{i,h+1}\^{t} \asto V_{i,h+1}^\pieps$. Thus from standard concentration results we know that
\begin{align*}
    \textstyle Q_{i,h}\^{t+1}(s,a) \asto r_{i,h}(s,a)\sum_{s'}P(s'|s,a) V_{i,h+1}^\pieps(s') = Q_{i,h}^\pieps(s,a)
\end{align*}

After proving $Q_{i,h}\^{t+1}(s,a)\asto Q_{i,h}^\pieps(s,a)$, the rest is the same as the proof of Lemma \ref{lemma:pi-eps-optimality}, where we can apply Lemma \ref{lemma:auxiliary-main} (set $X_t:= (\pi_h\^t(s), \xi_h\^{t}(s)), \!Q_t:= {Q_{i,h}\^t(s)}_{i=\!1}^n$ in the statement of Lemma \ref{lemma:auxiliary-main}) and get that $\pi_h\^t(s), \xi_h\^{t}(s)$ converges to $\pieps_h(\cdot,\cdot|s)$ in distribution and that
    \begin{align*}
     \textstyle   \sum_{\tau=1}^t\frac{Q^\pieps_{i,h}(s,\pi_h\^{\tau}(s))}{t}\asto \sum_{a}\pieps_h(a|s)Q_{i,h}^\pieps(s,a) = V^\pieps_{i,h}(s).
    \end{align*}
    Additionally, since  $Q_{i,h}\^{t}\asto Q_{i,h}^\pieps$, we also get $V_{i,h}\^{t}(s) =  \sum_{\tau=1}^t\frac{Q\^{\tau}_{i,h}(s,\pi_h\^{\tau}(s))}{t}\asto V^\pieps_{i,h}(s),$
    which completes the proof.
\end{proof}

\section{Auxiliaries}\label{apdx:auxiliaries}
\paragraph{Notations}
We first define the mixing time for the time-homogeneous and the non-time-homogeneous Markov chain. We borrow the definition from \citep{Paulin18}
\begin{defi}[Mixing time for time-homogeneous Markov chains]
    Let \\$X_1, X_2, X_3, \dots$ be a time-homogeneous Markov chain with transition kernel $P$ and stationary distribution $\pi$. Then $\tmix$, the mixing time of the chain, is defined by 
    \begin{equation}\label{eq:def-tmix}
        \begin{split}
           \textstyle d(t):= \sup_{\rho}\|P^t\rho - \pi\|_1,\quad 
            \tmix(\epsilon):=\{t: d(t)\le \epsilon\}.
        \end{split}
    \end{equation}
\end{defi}

\begin{defi}[Mixing time for Markov chains without assuming time homogeneity]
Let $X_1, X_2, X_3, \dots, X_N$ be a Markov chain with transition probabilities \\$P_1, P_2, P_3,\dots$. Let $L(X_{i+t}|X_i = x)$ be the conditional distribution of $X_{i+t}$ given $X_i = x$. Let us denote the minimal $t$ such that $L(X_{i+t}|X_i = x)$ and $L(X_{i+t}|X_i = y)$ are less than $\epsilon$ away in $\ell_1$ distance for every $1\le i\le N-t$ by $\taumix(\epsilon)$, that is, for $0 < \epsilon < 1$, let
\vspace{-5pt}
\begin{equation}\label{eq:def-tau-mix}
    \begin{split}
        \overline{d}(t)&\textstyle := \max_{1\le i\le N-t} \sup_{x,y} \|L(x_{i+t}|X_i=x),L(x_{i+t}|X_i=y)\|_1\\
        \taumix(\epsilon) &\textstyle := \min\{t\in\mathbb{N}: \overline{d}(t)\le \epsilon\}.
    \end{split}
\end{equation}
\end{defi}
\begin{lemma}\label{lemma:auxiliary-main-underlying}
    Suppose we have a random process $\{Q_t\}_{t=1}^{+\infty}$ such that
    $Q_t\asto q^\star,$ where $q^*$ is a constant vector (the random process doesn't necessarily need to be Markovian). And we have another random process $\{X_t\}_{t=1}^{+\infty}$, where $X_t$ is a random variable on domain $\cX$, which satisfies for any $\tau \ge 0$
    \begin{align}\label{eq:auxiliary-lemma-eq}
        \Pr(X_{t+1}|X_{1:t}, Q_{1:t+\tau}) = P_X(X_{t+1}|X_{t}, Q_t),
    \end{align}
     where the transition kernel $P_X(x'|x, Q)$ is continuous with respect to $Q$ for every $x', x \in\cX$. 

         Define $\pi_X$ as the stationary distribution of the transition kernel $P_{\!X}(\cdot|\cdot, \!q^*\!)$, i.e.
    $$\textstyle \pi_X(x') = \sum_{x} P_x(x'|x, q^*)\pi_X(x).$$
Then the following statements hold
    \begin{enumerate}
        \item[(1)] $\lim_{t\to+\infty}\Pr(X_t = x) = \pi_X(x).$
        \item[(2)]For any bounded function $g$, $ \frac{\sum_{\tau=1}^t g(X_\tau)}{t} \asto g^*$, where $g^*:= \bE_{X\sim \pi_X} g(X)$.
    \end{enumerate}
\begin{proof}
From the property of conditional probabilities and \eqref{eq:auxiliary-lemma-eq} we get
\begin{talign*}
    \Pr(X_{1:t\!+\!1}| Q_{1:t}) \!=\! \prod_{\tau\!=\!1}^t\!\Pr(X_{\tau\!+\!1} |X_{1:\tau}, \!Q_{1:t}) \!\Pr(X_1|Q_{1:t}) \!=\! \prod_{\tau \!=\! 1}^t \!P_X(X_{\tau+1} |X_\tau,\! Q_\tau) \!\Pr(X_1)
\end{talign*}
This means that given the event $\{Q_{t} = q_t, Q_{t-1} = q_{t-1}, \dots, Q_1 = q_1\}$, then conditional on this event, $\{X_\tau\}_{\tau=1}^t$ forms a non-homogeneous Markov chain with transition probability defined as $P_\tau(x'|x):= P_X(x'|x, q_\tau).$

Additionally, we also denote the transition kernel $P^*$ as $P^*(x'|x) = P_X(x'|x, q^*)$ and denote the $\epsilon$-mixing time of $P^*$ as $\tmix(\epsilon)$. We also denote the event $A:= \left\{\||P_t - P^*\|_1 \le \frac{\epsilon}{2\tmix\left(\frac{\epsilon}{2}\right)},~\forall t\ge T\right\}$. Since $Q_t \asto q^*$, and that $P_X(\cdot|\cdot, Q)$ is continuous with respect to $Q$, $P_t:=P_X(x'|x,Q_\tau)\asto P^*$. Thus, for every $\epsilon >0$, there exists a $T$ such that, $\Pr(A) \ge 1-\epsilon$.

    We first prove statement (2). Since conditional on the event $B:=\{Q_{t} = q_t, Q_{t-1} = q_{t-1}, \dots, Q_1 = q_1\}$, $\{X_\tau\}_{\tau=1}^t$ forms a non-homogeneous Markov chain with transition probabilities $P_\tau(x'|x):= P_X(x'|x, q_\tau)$. Applying  Corollary \ref{coro:auxiliary-mixing-time-for-non-homogeneous-Markov-chain} and Lemma \ref{lemma:concentration-non-homogeneous Markov chain}, we get that for any $t\ge T_1:= T + \tmix\left(\frac{\epsilon}{2}\right)$, 
    \begin{talign*}
        \Pr\left(\left|\frac{\sum_{\tau=T_1}^t g(X_\tau)}{t-T_1+1} - \bE\frac{\sum_{\tau=T_1}^t g(X_\tau)}{t-T_1+1}\right| \ge 2\epsilon  ~\Large|~ B\right) \le 2e^{-\frac{72\epsilon^2 t}{49M\taumix\left(\frac{1}{4}\right)}}.
    \end{talign*}
    Further, applying Corollary \ref{coro:auxiliary-convergence-of-mean} we get that if $B\subset A$, then for $t \ge T+ \frac{\epsilon}{2\tmix\left(\frac{\epsilon}{2}\right)}$
    \begin{talign*}
        \left|\bE\frac{\sum_{\tau=T_1}^t g(X_\tau)}{t-T_1+1}-g^* \right| \ge \epsilon 
    \end{talign*}
    Combining the above statement we get that for $t \ge T_1$
    \begin{talign*}
    &\quad \Pr\left(\left|\frac{\sum_{\tau=T_1}^t g(X_\tau)}{t-T_1+1} - g^*\right| \ge \epsilon  ~\Large|~ A,B\right) \\
    &
        \le \Pr\left(\left|\frac{\sum_{\tau=T_1}^t g(X_\tau)}{t-T_1+1} - \bE\frac{\sum_{\tau=T_1}^t g(X_\tau)}{t-T_1+1}\right| \ge 2\epsilon  ~\Large|~ A,B\right) \le 2e^{-\frac{72\epsilon^2 (t-T_1+1)}{49M\taumix\left(\frac{1}{4}\right)}},
    \end{talign*}
    \vspace{-10pt}
    \begin{talign*}
     \textup{and thus}\hspace{43pt} &\Pr\left(\left|\frac{\sum_{\tau=T_1}^t g(X_\tau)}{t-T_1+1} - g^*\right| \ge \epsilon  ~\Large|~ A\right)  \le 2e^{-\frac{72\epsilon^2 (t-T_1+1)}{49M\taumix\left(\frac{1}{4}\right)}}\\
        \Longrightarrow &\sum_{t\ge T_1}\Pr\left(\left|\frac{\sum_{\tau=T_1}^t g(X_\tau)}{t-T_1+1} - g^*\right| \ge \epsilon  ~\Large|~ A\right) < +\infty\\
        \Longrightarrow \textup{ there exists }& T_2 \textup{ such that }
        \Pr\left(\left|\frac{\sum_{\tau=T_1}^t g(X_\tau)}{t-T_1+1} \!-\! g^*\!\right| \!\le \!\epsilon, ~\forall  t \ge T_1 \!+\! T_2  ~\Large| A\!\right) \!\ge\! 1\!-\!\epsilon
    \end{talign*}
    Thus the following equality holds for any $\epsilon > 0$
    \begin{talign*}
        &\Pr\left(\left|\frac{\sum_{\tau=T_1}^t g(X_\tau)}{t-T_1+1} - g^*\right| \le \epsilon, ~\forall  t \ge T_1 + T_2\right)\\
        &\ge \Pr(A)\Pr\left(\left|\frac{\sum_{\tau=T_1}^t g(X_\tau)}{t-T_1+1} - g^*\right| \le \epsilon, ~\forall  t \ge T_1 + T_2   ~\Large|~ A\right) \ge (1-\epsilon)^2.
    \end{talign*}
    Note that this inequality is equivalent to $\frac{\sum_{\tau=T_1}^t g(X_\tau)}{t-T_1+1} \asto g^*,$ which is also equivalent to $\frac{\sum_{\tau=1}^t g(X_\tau)}{t} \asto g^*,$ which proves statement (2).

    We now prove statement (1). By setting $g(x) = \mathbf{1}\{x = x_0\}$ which is easier to prove given the fact from Corollary \ref{coro:auxiliary-convergence-of-mean} that under event $A$, for $t \ge T+ \frac{\epsilon}{2\tmix\left(\frac{\epsilon}{2}\right)}$
    \begin{talign*}
        |\bE g(X_t) - g^*|\le \epsilon, ~~\bE g(X_t) = \Pr(X_t = x_0),~~ g^* = \pi_X(x_0)
    \end{talign*}
    thus for $t \ge T+ \frac{\epsilon}{2\tmix\left(\frac{\epsilon}{2}\right)}$
    \begin{talign*}
        |\Pr(X_t = x_0) - \pi_X(x_0)| \le P(A) \epsilon + (1-P(A)) = \epsilon + (1-\epsilon)\epsilon \le 2\epsilon
    \end{talign*}
    Let $\epsilon\!\to\! 0$, we have that $\lim_{t\to\infty} \Pr(X_t \!=\! x_0) \!=\! \pi_X(x_0)$, which proves statement (1).
\end{proof}
\end{lemma}

\begin{lemma}\label{lemma:auxiliary-main}
     For random variables $Q_{h}\^{t}(s,\cdot)$ that almost surely converges to \\$Q_{h}^{\pieps_{h+1:H}}(s,\cdot)$, we can define the following random processes
     \begin{align}
   a_h\^{t+1}(s),\xi_h\^{t+1}(s) &\sim\Peps(\cdot|a_{h}\^{t}(s),\xi_h\^{t}(s); \brac{Q_{i,h}\^{t}(s,\cdot)}) \label{eq:algorithm-block-1}\\
   {a'}_{h}\^{t+1}(s),{\xi'}_h\^{t+1}(s) &\sim\Peps(\cdot|{a'}_{h}\^{t}(s),{\xi'}_h\^{t}(s); \brac{Q_{i,h}^{\pieps_{h+1:H}}(s,\cdot)}) \label{eq:auxiliary-block-1}
\end{align}
Define the stationary distribution of \eqref{eq:algorithm-block-1} and \eqref{eq:auxiliary-block-1} as $\pieps_h(a,\xi|s),{\pieps'}_h(a,\xi|s)$, then we have that
\begin{enumerate}
    \item $\pieps_h(a,\xi|s)={\pieps'}_h(a,\xi|s)$
    \item For any bounded function, e.g. $ Q_{i,h}^{\pieps_{h+1:H}}(s,\cdot)$, we have that
    \vspace{-5pt}
    \begin{align*}
    \textstyle    \frac{1}{t}\sum_{\tau=1}^t Q_{i,h}^{\pieps_{h+1:H}}(s, a_h\^{t}(s)) \asto  \bE_{a\sim {\pieps}'_h} Q_{i,h}^{\pieps_{h+1:H}}(s, a)
    \end{align*}
\end{enumerate}

\begin{proof}

The lemma is a direct corollary of Lemma \ref{lemma:auxiliary-main-underlying}. (By setting\\
{\small $X_t:= (a_h\^t(s), \xi_h\^{t}(s)), Q_t:= \brac{Q_{i,h}\^t(s,\cdot)}, q^\star := \brac{Q_{i,h}^{\pieps_{h+1:H}}(s,\cdot)}, g := Q_{i,h}^{\pieps_{h+1:H}}(s,\cdot)$} \\in the statement of Lemma \ref{lemma:auxiliary-main-underlying})  
\end{proof}
\end{lemma}

\begin{lemma}\label{lemma:auxiliary-non-homogeneous Markov chain}
    Suppose there's a probability transition kernel $P^*$ that defines an Ergodic Markov chain with mixing time $\tmix(\epsilon) < \infty$, where $\tmix(\epsilon)$ is defined as in \eqref{eq:def-tmix}. We also denote the stationary distribution of $P^*$ as $\pi^*$, i.e. $P^*\pi^* = \pi^*$. Suppose there's a sequence of probability transition kernel $\{P_t\}_{t=1}^{+\infty}$ satisfies $\|P_t - P^*\|_1 \le \frac{\epsilon}{2\tmix\left(\frac{\epsilon}{2}\right)}$, then for any $t \ge \tmix\left(\frac{\epsilon}{2}\right)$ and any initial distribution $\rho$, we have that
    \begin{talign*}
        \|P_t P_{t-1}\cdots P_1 \rho - \pi^*\|_1 \le \epsilon.
    \end{talign*}
    \vspace{-10pt}
    \begin{proof}
        Let $\pi_t:= P_t P_{t-1}\cdots P_1 \rho $, then we get
        \begin{talign*}
            \pi_{t+1} &= P_t\pi_t = P^*\pi_t - (P^* - P_t)\pi_t\\
            \Longrightarrow~~ \pi_{t+1} - \pi^* &= P^*(\pi_t - \pi^*)- (P^* - P_t)\pi_t\\
            &= (P^*)^2(\pi_{t-1}-\pi^*) - P^*(P^*-P_{t-1})\pi_{t-1} - (P^*-P_t)\pi_t\\
            &= (P^*)^{\tmix\left(\frac{\epsilon}{2}\right)}\pi_{t+1-\tmix\left(\frac{\epsilon}{2}\right)} - \pi^* + \sum_{\tau= t-\tmix\left(\frac{\epsilon}{2}\right)}^t (P^*)^{t-\tau}(P^* - P_\tau) \pi_\tau\\
        \end{talign*}
        Since $P^*$ defines a Markov chain with mixing time $\tmix$, from the definition of mixing time we have \vspace{-10pt}
        \begin{align*}
           \textstyle \|(P^*)^{\tmix\left(\frac{\epsilon}{2}\right)}\pi_{t+1-\tmix\left(\frac{\epsilon}{2}\right)} - \pi^*\|_1 \le \frac{\epsilon}{2}.\vspace{-10pt}
        \end{align*}
        \begin{align*}
          \textstyle \textup{Further, }~~~~\| \sum_{\tau= t-\tmix\left(\frac{\epsilon}{2}\right)}^t (P^*)^{t-\tau}(P^* - P_\tau) \pi_\tau \|_1 \le \sum_{\tau= t-\tmix\left(\frac{\epsilon}{2}\right)}^t \|(P^* - P_\tau)\|_1 \le \frac{\epsilon}{2}
        \end{align*}
        Combining the inequality we get $\|\pi_{t+1} -\pi^*\|_1 \!\le\! \epsilon$,
        which completes the proof.
    \end{proof}
\end{lemma}
\begin{coro}\label{coro:auxiliary-convergence-of-mean}
    For a bounded function $g$ and a non-homogeneous Markov chain with transition kernel $\{P_t\}_{t=1}^\infty$ defined in Lemma \ref{lemma:auxiliary-non-homogeneous Markov chain} the random variable $X_1, X_2, \dots$ generated from the Markov chain satisfies
    \begin{align*}
        \textstyle|\bE g(X_t) - g^*|\le M\epsilon, ~~\forall ~ t\ge \tmix\left(\frac{\epsilon}{2}\right).
    \end{align*}
    where $M$ is defined as $M:=\sup_{x,y\in \cX} g(x) - g(y)$ and $g^* = \bE_{x\sim \pi^*}g(x)$
\end{coro}

\begin{coro}\label{coro:auxiliary-mixing-time-for-non-homogeneous-Markov-chain}
    The mixing time $\taumix(\epsilon)$ for the non-homogeneous Markov chain with transition kernel $\{P_t\}_{t=1}^\infty$ defined in Lemma \ref{lemma:auxiliary-non-homogeneous Markov chain} is finite for every $\epsilon > 0$. Specifically, we have $\taumix(2\epsilon) \le \tmix\left(\frac{\epsilon}{2}\right).$
\end{coro}

\begin{lemma}[Corollary 2.11 in \cite{Paulin18}]\label{lemma:concentration-non-homogeneous Markov chain}
    For any potentially non-homogeneous discrete time Markov chain with finite mixing time $\taumix(\epsilon) < \infty$, the random variable $X_1, X_2, \dots$ generated from the Markov chain satisfies
    \vspace{-5pt}
    \begin{align*}
       \textstyle  \Pr\left(\left|\frac{\sum_{\tau=1}^t g(X_\tau)}{t} - \bE\frac{\sum_{\tau=1}^t g(X_\tau)}{t}\right| \ge \epsilon \right) \le 2e^{-\frac{18\epsilon^2 t}{49M\taumix\left(\frac{1}{4}\right)}},
    \end{align*}
    where $g$ is a bounded function and $M$ is defined as $M:=\sup_{x,y\in \cX} g(x) - g(y)$.
\end{lemma}






%% file: Sections/appendix_numerics.tex
\section{Another Numerical Example}\label{apdx:another-numerical-example}

To have a better understanding of Theorem \ref{theorem:stochastic-stable-global-optimality}, we provide an additional numerical example. The stochastic game that we consider is an identical interest, two-player, two-action game with two stages $h=1,2$. The game captures the process of two players collaboratively digging treasure at two different locations $0$ and $1$ (Figure \ref{fig:numerics} (left)). Each location has a shallow level and a deep level, and it requires the players to collaboratively dig at the same location to make progress. Location $0$ has reward $1$ at the shallow level and $0.5$ at the deep level; location $1$ has reward $0$ at the shallow level and $2$ at the deep level. The process can be summarized as a stochastic game as in Fig \ref{fig:identical-interest-setup}. For $h=1$, there's only one state and thus the stage reward is given by a 2 by 2 reward matrix with nonzero reward of value 1 only when $a = (0,0)$. The transition to the second stage follows the following transition rule, if both players 1 and 2 choose to dig at location $0$, i.e. $a_1 = a_2 = 0$, then the game will transit to state $s=A$, where they arrive at the deep level for location 0 and shallow level at location 1. Similarly, if $a_1 = a_2 = 1$ then the game will transit to $s=B$, where they arrive at the deep level for location $1$. Otherwise, if the players fail to agree on a location, i.e., $a=(0,1) \textup{or} (1,0)$ they make no progress ($s=O$) and the reward will remain the same as stage $h=1$ in this case.

\begin{figure}[tb]
\centering
\begin{small}
\input{tikfigures/2x2example_identical_interest}
\end{small}
\caption{Game schematic (states, transitions, payoffs).}
\label{fig:identical-interest-setup}
\end{figure}

In this game, there are two strict NEs, agreeing on location $0$, i.e. both players choosing $a_{i,h} = 0$ for $h=1,2$, or agreeing on location $1$, where $a_{i,h}=1$. Note that the second NE is also the global optimal policy that gives a total reward $2$ while the first only gives reward $1.5$. Theorem \ref{theorem:stochastic-stable-global-optimality} suggests that for log-linear learning, the stochastically stable policy is the global optimal policy, thus it should converge to the second instead of the first NE. Our numerical result (Figure \ref{fig:numerics} (right)) also corroborates the theoretical finding, where we find that the empirical frequency of selecting location $1$ will converge to a value close to $1$.

\begin{figure}
  \centering
\includegraphics[width=0.35\textwidth]{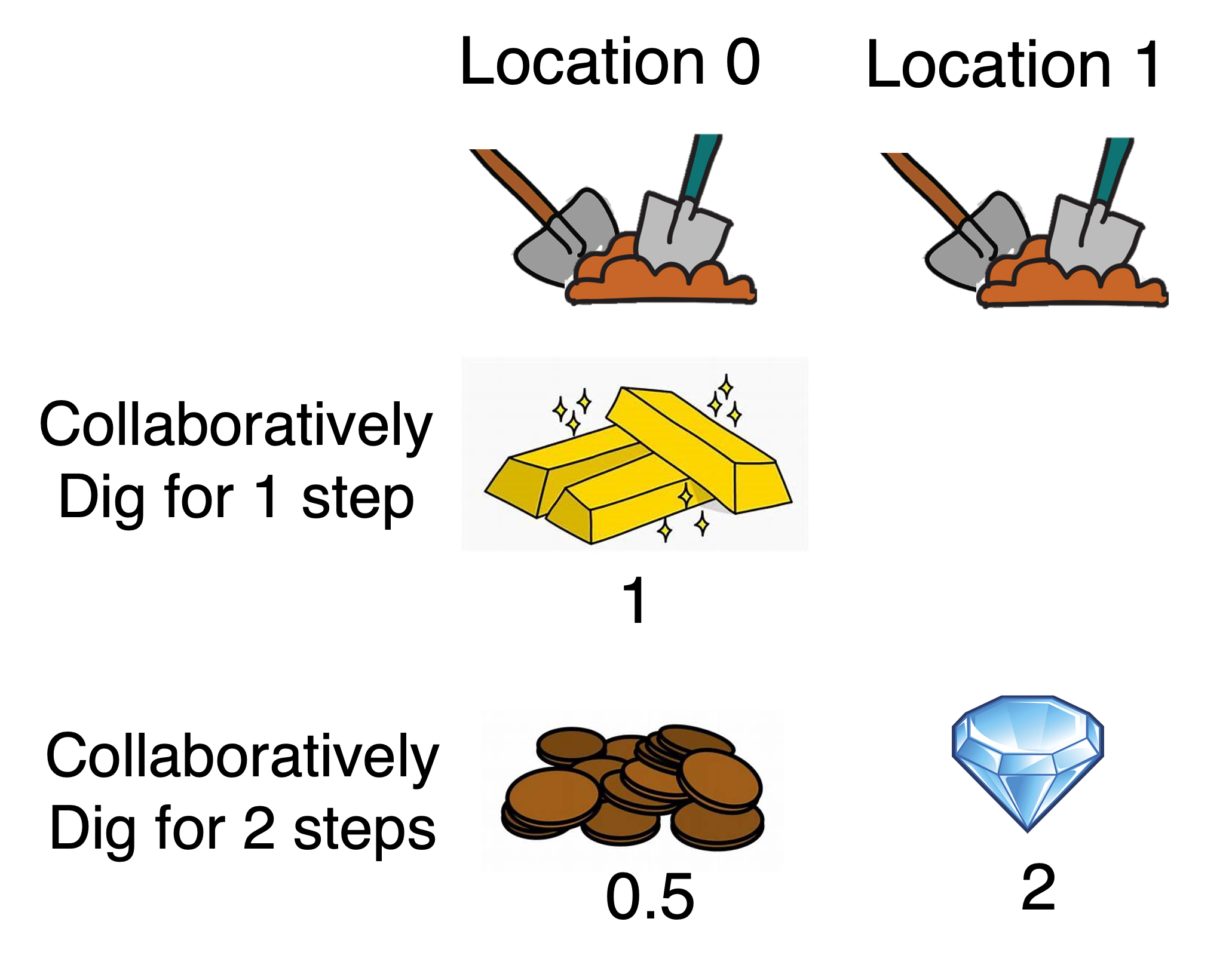}
\hspace{20pt}
\includegraphics[width=0.45\textwidth]{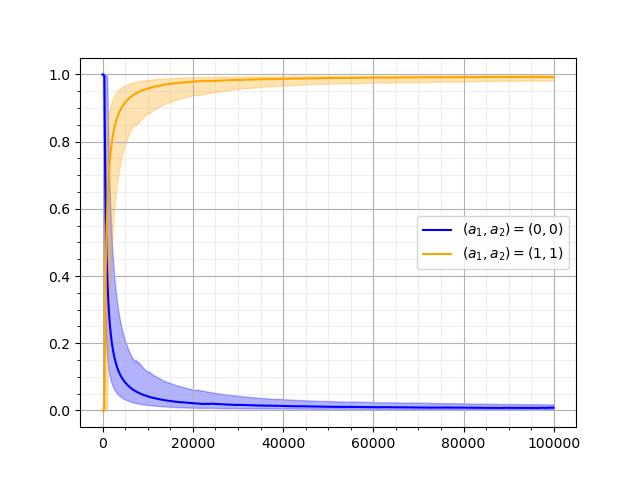}
  \vspace{-10pt}
  \caption{\small (Left) Illustration of 2-stage the stochastic game. (Right) Numerical result of running Algorithm \ref{alg:unified-framework} with log-linear learning rule (with $\epsilon = 10^{-5}$). The plot demonstrates the evolution of $\pi\^{t}(a|s)$ (estimated using empirical averaging), where $s$ is the initial state and the yellow line plots the curve for $a = (1,1)$ and the blue line $a=(0,0)$. The shaded areas are the $60\%$ confidence interval calculated by 100 runs.}
  \label{fig:numerics}
\end{figure}

%% file: tikfigures/2x2example_identical_interest.tex
\setlength{\tabcolsep}{4pt}
\vspace{4pt}
\begin{center}
{\small 
\begin{tabular}{|c|c|c|}
\hline 
  & $\displaystyle a_{2} \ =\ 0$ & $\displaystyle a_{2} =1$ \\
\hline 
 $\displaystyle a_{1} =0$ & $\displaystyle 1$ & $\displaystyle 0$ \\
\hline 
 $\displaystyle a_{1} =1$ & $\displaystyle 0$ & $\displaystyle 0$ \\
 \hline
\end{tabular}}

\vspace{2pt}
{\small $\displaystyle h=1$}
\end{center}
\vspace{4pt}

{\small 

\tikzset{every picture/.style={line width=0.75pt}} 
\begin{center}
    
\begin{tikzpicture}[x=0.75pt,y=0.75pt,yscale=-0.7,xscale=0.6]

\draw    (353,2.8) -- (137.93,59.49) ;
\draw [shift={(136,60)}, rotate = 345.23] [color={rgb, 255:red, 0; green, 0; blue, 0 }  ][line width=0.75]    (10.93,-3.29) .. controls (6.95,-1.4) and (3.31,-0.3) .. (0,0) .. controls (3.31,0.3) and (6.95,1.4) .. (10.93,3.29)   ;
\draw    (353,2.8) -- (353,59) ;
\draw [shift={(353,61)}, rotate = 270] [color={rgb, 255:red, 0; green, 0; blue, 0 }  ][line width=0.75]    (10.93,-3.29) .. controls (6.95,-1.4) and (3.31,-0.3) .. (0,0) .. controls (3.31,0.3) and (6.95,1.4) .. (10.93,3.29)   ;
\draw    (353,2.8) -- (557.08,64.22) ;
\draw [shift={(559,64.8)}, rotate = 196.75] [color={rgb, 255:red, 0; green, 0; blue, 0 }  ][line width=0.75]    (10.93,-3.29) .. controls (6.95,-1.4) and (3.31,-0.3) .. (0,0) .. controls (3.31,0.3) and (6.95,1.4) .. (10.93,3.29)   ;

\draw (100,23) node [anchor=north west][inner sep=0.75pt]    {$a_{1} =a_{2} =0$};
\draw (485,23) node [anchor=north west][inner sep=0.75pt]    {$a_{1} =a_{2} =1$};
\draw (358,23) node [anchor=north west][inner sep=0.75pt]    {Else};
\end{tikzpicture}
\end{center}
}
\vspace{-2pt}
\begin{center}
{\small $\displaystyle  \begin{array}{ccc}
s=A   \ \ \ \ \ \ \ \ \ \ \ \ \ \ \ \ \ \ \ \ \ \ \ \ \ \ \ \ s=O\ \ \ \ \ \ \ \ \ \ \ \ \ \ \ \ \ \ \ \ \ \ \ \ \ \ \ \ s =B 
\end{array}$}

{\small 
\begin{tabular}{|c|c|c|}
\hline 
  & $\displaystyle a_{2}  = 0$ & $\displaystyle a_{2} =1$ \\
\hline 
 $\displaystyle a_{1} =0$ & $\displaystyle 0.5$ & $\displaystyle 0$ \\
\hline 
 $\displaystyle a_{1} =1$ & $\displaystyle 0$ & $\displaystyle 0$ \\
 \hline
\end{tabular} \ \ \ 
\begin{tabular}{|c|c|c|}
\hline 
  & $\displaystyle a_{2}  = 0$ & $\displaystyle a_{2} =1$ \\
\hline 
 $\displaystyle a_{1} =0$ & 1 & $\displaystyle 0$ \\
\hline 
 $\displaystyle a_{1} =1$ & $\displaystyle 0$ & $\displaystyle 0$ \\
 \hline
\end{tabular} \ \ \
\begin{tabular}{|c|c|c|}
\hline 
  & $\displaystyle a_{2} \ =\ 0$ & $\displaystyle a_{2} =1$ \\
\hline 
 $\displaystyle a_{1} =0$ & $\displaystyle 1$ & $\displaystyle 0$ \\
\hline 
 $\displaystyle a_{1} =1$ & $\displaystyle 0$ & $\displaystyle 2$ \\
 \hline
\end{tabular}}

\vspace{2pt}
{\small $\displaystyle h=2$}
\end{center}